\RequirePackage[T1]{fontenc}
\documentclass[11pt]{article}
\usepackage[margin=1in]{geometry}
\usepackage[nocompress]{cite}
\usepackage{ifpdf}
\usepackage{graphicx}

\usepackage{xspace,comment,calc,xcolor}
\usepackage[cmex10]{amsmath}
\usepackage{amssymb,amsthm,amsfonts,amstext}

\usepackage[shortlabels,inline]{enumitem}
\usepackage{multirow}
\usepackage{algorithm}
\usepackage{algpseudocodex}
\usepackage{dsfont}
\newcommand{\mathbbm}[1]{\mathds{#1}}
\usepackage{caption}
\usepackage{url}

\usepackage[backref=page,colorlinks=true,urlcolor=blue,
  citecolor=red,linkcolor=blue,linktocpage,pdfpagelabels,
  bookmarksnumbered,bookmarksopen,hidelinks]{hyperref} 
\usepackage[hyperpageref]{backref}
\hypersetup{
	bookmarks=true,
	bookmarksnumbered=true,
	bookmarksopen=true,
    hypertexnames=false
}

\usepackage[utf8]{inputenc}
\usepackage{newunicodechar}
\newunicodechar{ȩ}{\c{e}}
\newunicodechar{–}{\textendash}

\newcommand{\R}{\ensuremath{\mathbb R}}
\newcommand{\Z}{\ensuremath{\mathbb Z}}

\newcommand{\eps}{\varepsilon}
\newcommand{\mint}{\ensuremath{\mathsf{int}}}
\newcommand{\Mint}{M_{\mint}}
\newcommand{\chimax}{\chi_{\max}}

\newcommand{\lpmat}{\text{LP}_{\text{mat}}}

\newcommand{\I}{\ensuremath{\mathcal{I}}}
\newcommand{\M}{\ensuremath{\mathcal{M}}}
\newcommand{\Pc}{\ensuremath{\mathcal P}}
\newcommand{\ve}{\varepsilon}

\newcommand{\dt}{\ensuremath{\delta}}
\newcommand{\Dt}{\ensuremath{\Delta}}
\newcommand{\poly}{\operatorname{poly}}
\newcommand{\w}{\omega}

\newcommand{\OPT}{\ensuremath{\mathit{OPT}}}

\newcommand{\sm}{\ensuremath{\setminus}} 
\newcommand{\es}{\ensuremath{\emptyset}}
\newcommand{\sse}{\subseteq} 
\newcommand{\assign}{\leftarrow}

\newcommand{\np}{{\em NP}\xspace} 
\newcommand{\nphard}{\np-hard\xspace}

\newcommand{\mat}{M}
\newcommand{\gset}{U}
\newcommand{\inds}{\I}
\newcommand{\rk}{r}
\newcommand{\bmat}{\overline\mat}

\newcommand{\binds}{\overline\inds}

\newcommand{\bN}{\overline N}
\newcommand{\bI}{\overline I}
\newcommand{\matref}{\mat'}
\newcommand{\gsetref}{\gset'}
\newcommand{\indsref}{\inds'}

\newcommand{\mbase}{\widetilde\mat}
\newcommand{\gbase}{\widetilde\gset}
\newcommand{\indbase}{\widetilde\inds}
\newcommand{\losz}{\textsc{LOSZ}\xspace}

\newcommand{\bon}{\ensuremath{\mathbbm{1}}}
\newcommand{\baseP}{\Pc_{B}}
\newcommand{\ceil}[1]{\ensuremath{\bigl\lceil#1\bigr\rceil}}

\newenvironment{proofof}[1]{\begin{proof}[Proof of {#1}]}{\end{proof}}

\makeatletter
\def\@citex[#1]#2{\leavevmode
  \let\@citea\@empty
  \@cite{\@for\@citeb:=#2\do
    {\@citea\def\@citea{,\penalty\@m\ }%
\edef\magic##1{\let##1\expandafter\noexpand\csname bibalias@\@citeb\endcsname}%
\magic\tmp \ifx\tmp\relax\else \let\@citeb\tmp\fi
     \edef\@citeb{\expandafter\@firstofone\@citeb\@empty}%
     \if@filesw\immediate\write\@auxout{\string\citation{\@citeb}}\fi
     \@ifundefined{b@\@citeb}{\hbox{\reset@font\bfseries ?}%
       \G@refundefinedtrue
       \@latex@warning
         {Citation `\@citeb' on page \thepage \space undefined}}%
       {\@cite@ofmt{\csname b@\@citeb\endcsname}}}}{#1}}
\def\bibalias#1#2{\expandafter\def\csname bibalias@#1\endcsname{#2}}
\makeatother

\newcommand{\IfThen}[1]{\algorithmicif\ #1 \algorithmicthen}

\newtheorem{theorem}{Theorem}[section]
\newtheorem{lemma}[theorem]{Lemma}
\newtheorem{corollary}[theorem]{Corollary}

\newtheorem{claim}[theorem]{Claim}

\newtheorem{fact}[theorem]{Fact}
\newtheorem{openquestion}{Question}

{\theoremstyle{definition}

\newtheorem*{remarks}{Remarks}
\newtheorem{definition}[theorem]{Definition}

\newtheorem{conjecture}{Conjecture}

}

\usepackage[colorinlistoftodos,prependcaption,textsize=tiny]{todonotes}

\usepackage[noabbrev]{cleveref}

\begin{document}

\title{Approximation Algorithms for Matroid-Intersection Coloring with Applications to
  Rota's Basis Conjecture%
\thanks{S. Arndt, B. Moseley, and K. Pruhs were partially
funded by NSF grant CCF-2550399. C. Swamy was funded by an NSERC Discovery grant.}
}

\author{%
Stephen Arndt\thanks{Tepper School of Business, Carnegie Mellon University, Pittsburgh, USA. Email: \texttt{sarndt@andrew.cmu.edu}.}
\quad Benjamin Moseley\thanks{Tepper School of Business, Carnegie Mellon University, Pittsburgh, USA. Email: \texttt{moseleyb@andrew.cmu.edu}.}
\quad Kirk Pruhs\thanks{Dept. of Computer Science, University of Pittsburgh, Pittsburgh, USA. Email: \texttt{kirk@cs.pitt.edu}.}
\\[0.6em] Chaitanya Swamy\thanks{Combinatorics \& Optimization, University of Waterloo, Waterloo, CANADA. Email: \texttt{cswamy@uwaterloo.ca}.}
\quad Michael Zlatin\thanks{Dept. of Computer Science, Pomona College, Claremont, USA. Email: \texttt{michael.zlatin@pomona.edu}.}
}
\date{}

\maketitle

\bibalias{LinharesOSZ20}{swamy_algorithm}
\bibalias{AharoniBGK25}{AharoniBergerGuoKotlar2025}
\bibalias{partredn1}{part_decomp_1}
\bibalias{partredn2}{part_decomp_2}
\bibalias{partredn3}{part_decomp_gammoid}
\bibalias{BercziS21}{BercziSchwarcz2021}
\bibalias{MontgomeryS25}{montgomery2025}

\begin{abstract}
We study algorithmic matroid intersection coloring. Given $k$ matroids on a common ground set $U$ of $n$ elements, the goal is to partition $U$ into the fewest number of color classes, where each color class is independent in all matroids. It is known that $2\chimax$ colors suffice to color the intersection of two matroids, $(2k-1)\chimax$ colors suffice for general $k$, where $\chimax$ is the maximum chromatic number of the individual matroids. However, these results are non-constructive, leveraging techniques such as topological Hall's theorem and Sperner's Lemma.

We provide the first polynomial-time algorithms to color two or more general matroids where the approximation ratio depends only on $k$ and, in particular, is independent of $n$. For two matroids, we constructively match the $2\chimax$ existential bound, yielding a 2-approximation for the Matroid Intersection Coloring problem. For $k$ matroids we achieve a $(k^2-k)\chimax$ coloring, which is the first $O(1)$-approximation for constant $k$. Our approach introduces a novel matroidal structure we call a \emph{flexible decomposition}. We use this to formally reduce general matroid intersection coloring to graph coloring while avoiding the limitations of partition reduction techniques, and without relying on non-constructive  topological machinery.

Furthermore, we give a {\em fully polynomial randomized approximation scheme} (FPRAS) for coloring the intersection of two matroids when $\chimax$ is large. This yields the first polynomial-time constructive algorithm for an asymptotic variant of Rota's Basis Conjecture. This constructivizes Montgomery and Sauermann's recent asymptotic breakthrough and generalizes it to arbitrary matroids.
\end{abstract}

\section{Introduction} \label{intro}

Matroids are fundamental discrete objects that arise in a variety of settings in
combinatorics and optimization.  Consequently,  matroids and matroid-optimization problems have
 been extensively studied in the combinatorics, combinatorial-optimization,
and Theoretical Computer Science (TCS) literature.

We investigate a natural and well-studied matroid-optimization problem called
matroid-intersection coloring.
In an instance of {\em $k$-matroid-intersection coloring}, we are given  
$k$ matroids%
\footnote{A matroid is a tuple $\mat=(\gset,\inds)$, where $\gset$ is a
finite set and $\inds\sse 2^{\gset}$ satisfies: (i) $\es\in\inds$; (ii) 
$B\in\inds, A\sse B \Rightarrow A\in\inds$; and (iii) if $A,B\in\inds$, $|A|<|B|$, then
$\exists e\in B-A$ such that $A\cup\{e\}\in\inds$. If $\mat$ satisfies (i), (ii), we call
$\mat$ an {\em independence system}.}
$M_1=(U,\mathcal{I}_1),\ldots,M_k=(U,\mathcal{I}_k)$ on a common ground set, and the goal
is to cover $\gset$ using the minimum number of {\em common independent sets}. That is,
a feasible solution consists of sets $S_1,\ldots,S_q$, where $\bigcup_{c=1}^q S_c=\gset$,
and each $S_c$ is independent in the intersection of the $k$ matroids, 
which is defined as the independence system
$\Mint= \bigcap_{i=1}^k\mat_i=\bigl(\gset,\,\inds_{\mint}:=\bigcap_{i=1}^k\inds_i\bigr)$. 
Such a solution is called a {\em $q$-coloring} of $\mat_1,\ldots,\mat_k$, 
where we interpret elements in each $S_c$
as being ``colored'' with color $c$; we seek a solution using the fewest number of colors.
The optimal value is called the {\em chromatic number} of $\Mint$ and denoted $\chi(\Mint)$.
While we have phrased matroid-intersection coloring as a covering problem, observe that
one can always transform a solution to have pairwise-disjoint color-sets, so equivalently
the goal is to partition $\gset$ into the smallest number of common independent sets.
As is standard in matroid optimization, we assume that each input
matroid is specified via an independence oracle or a rank oracle. 

The chromatic number of a single matroid $\mat=(\gset,\inds)$ with rank function $\rk_{\mat}$ can be
computed efficiently via a reduction to matroid
intersection~\cite{edmonds1968matroid,schrijver_book}, which also yields 
the formula $\chi(\mat)=\max_{S\sse\gset}\ceil{\frac{|S|}{\rk_{\mat}(S)}}$, but 
matroid-intersection coloring quickly becomes intractable for $k\geq 2$.
It is known that even with two matroids, the problem %
is \nphard~\cite{BercziS21,gmpm_hard,horsch2024rainbow}, and  
computing $\chi(\Mint)$ exactly requires an exponential number of independence-oracle 
calls~\cite{BercziS21}. It is therefore natural to consider approximation algorithms, 
and approximation results in the literature are typically stated relative to the natural
lower bound $\chimax:=\max_{i=1,\ldots,k}\chi(\mat_i)$.

Matroid-intersection coloring has been studied both in the combinatorics
literature~\cite{ab06,AharoniBGK25,MontgomeryS25} and the combinatorial-optimization and
TCS literature~\cite{partredn1,partredn2,partredn3,arndt2025}. The current state of
affairs can be summarized as follows. 
The combinatorics literature has obtained various strong 
existential results %
via {\em nonconstructive}%
\footnote{Since we are considering problems with a finite solution space, by
``nonconstructive'' we mean not readily amenable to conversion to a polynomial-time
algorithm.}  
means. %
The optimization and TCS literature obtains constructive results (i.e.,  
polytime algorithms), %
but these apply only to certain ``combinatorial'' matroids (not general matroids). 
For general matroids, %
the only known constructive result prior to our work ({\em even for $k=2$}) is an 
$O(k\log n)$-approximation algorithm (i.e., an $O(k\log n)\cdot\chi(\Mint)$-coloring)
using the greedy algorithm for set cover.%
\footnote{The greedy step translates to solving a $k$-matroid-intersection problem,
which admits an $O(k)$-approximation.}

As is evident from this discussion, there is a {\em significant gap} between 
the nonconstructive results and constructive guarantees, which persists
even for the case of two matroids. 
{\em Bridging this gap is the chief motivating goal of our work}, and we approach this 
research goal by considering the following two well-otivated research directions. 

\begin{enumerate}[label=$\bullet$, topsep=0.5ex, itemsep=0.5ex, leftmargin=0pt,
widest=$\bullet$, itemindent=*] 
\item {\bf Constructive Guarantees for Matroid-Intersection Coloring.}
The state-of-the-art guarantees for matroid-intersection coloring are given by the following
two nonconstructive results.%
\footnote{We state here the guarantees in terms of $\chimax$; with two matroids, 
a more refined, still nonconstructive, bound
$\chi(\mat_1\cap\mat_2)\leq\chi(\mat_1)+\chi(\mat_2)$ was recently obtained~\cite{BergerGuo2025}.}  

\begin{theorem}[\cite{ab06}] \label{thm:nonconstructive2} \label{twochimax}
For any two matroids $M_1$ and $M_2$, we have
$\chi(M_1\cap M_2)\leq 2\chimax$.
\end{theorem}

\begin{theorem}[\cite{AharoniBergerGuoKotlar2025}]
\label{thm:nonconstructive1}
For $k$ (general) matroids $M_1,\ldots,M_k$, we have
$\chi\!\left(\Mint\right)\le(2k-1)\chimax$.
\end{theorem}

The proofs of Theorem~\ref{thm:nonconstructive2}
and Theorem~\ref{thm:nonconstructive1} are topological, using the notion of homotopical
connectivity and relying on 
Sperner's lemma to deduce the existence of a desired coloring. As Sperner's lemma is
essentially complete for the complexity class PPAD, making this approach constructive
would involve surmounting well-known complexity-theoretic hurdles. 

Prior algorithmic results focus on the setting where %
all matroids~\cite{partredn1,partredn2,partredn3}, or all but one of the
matroids~\cite{arndt2025}, are partition matroids. 
These approaches extend to standard ``combinatorial'' matroids (i.e., transversal matroids,
gammoids, laminar matroids, graphic matroids), as one can reduce to partition matroids
at the loss of a small factor in the coloring number~\cite{partredn1,partredn2,partredn3}. 
But this partition-reduction approach cannot be extended to general matroids because there
are matroids where one cannot ``simplify the independence-structure'' and move to a
partition matroid without incurring a {\em significant} blow-up in the chromatic
number~\cite{imposs_1,imposs_2}. In particular, the chromatic number of the partition matroid must 
necessarily increase as a function of $n=|\gset|$. 
 In a sense, the work of \cite{arndt2025}, who handle the case where one of the matroids is
a general matroid
can be seen as a limit of this partition-reduction approach, and as the authors of~\cite{arndt2025} indicate, this approach is not amenable to handle the setting of even {two} general matroids. 

This preface leads to the first question that we explore in this work.

\medskip

\qquad 
\begin{minipage}{0.9\linewidth}
\begin{openquestion} \label{openquestion1} \label{qn1}
Can we devise polynomial-time algorithms for $k$-matroid intersection coloring (with general
matroids) whose guarantees match, or nearly match, those in
Theorems~\ref{thm:nonconstructive2} and~\ref{thm:nonconstructive1}?  
\end{openquestion}
\end{minipage}

\smallskip

\item {\bf Asymptotic Rota's Basis Conjecture.}
The second research direction we consider involves one of the most-prominent problems in
matroid-intersection coloring, namely {\em Rota's Basis Conjecture}, which is the
following famous conjecture regarding rearranging the bases of a matroid. 

\begin{conjecture}[{\bf Rota's Basis Conjecture}] \label{rota}
Let $\mat_1=(\gset,\inds_1)$ be a rank-$r$ matroid where $\gset$ can be partitioned into $r$ 
disjoint bases $B^1, B^2, \dots, B^r$. %
Then, there exist $r$ disjoint bases $T_1,\ldots,T_r$ such that $|T_i\cap B^j|=1$ for all
$i,j=1,\ldots,r$. The bases $T_i$ are sometimes called rainbow bases, or transversal bases.
\end{conjecture}

In the language of matroid-intersection coloring, %
taking $\mat_2=(\gset,\inds_2)$ to be the partition matroid with parts $B^1,\ldots,B^r$
and capacity $1$ on each part, Rota's Basis Conjecture asserts that 
$\chi(\mat_1\cap\mat_2)=r$. 
Note that in this setting $\chimax=\chi(M_1)=\chi(M_2)=r=\sqrt{n}$, where $n=|\gset|$.

Since Rota made this conjecture in 1989, it has received significant attention (see, e.g.,
\cite{HUANG1994225,Drisko1997OnTN,Glynn,MontgomeryS25} and the references therein).
It was the subject of the 12th Polymath project~\cite{rota_polymath},
which improved a $2\chimax$ upper bound on $\chi(\mat_1\cap\mat_2)$ 
to $2\chimax-2$.
The conjecture %
holds for some special cases~\cite{montgomery2025}, 
perhaps most notably for real-representable matroids when the rank is within one of an odd
prime~\cite{HUANG1994225,Drisko1997OnTN,Glynn}. 
Very recently,  Montgomery and Sauermann~\cite{MontgomeryS25} proved that Rota's basis
conjecture holds asymptotically, giving the current-best bounds on $\chi(M_1\cap M_2)$.
They prove the following.  

\begin{theorem}[\cite{montgomery2025}]
\label{thm:montgomery2025}
For any $\varepsilon>0$, and any instance
of Rota's Basis Conjecture where $\chi(\mat_1)$ (or equivalently the rank $r$) 
is sufficiently large relative to $\varepsilon$, we have %
$\chi(M_1\cap M_2)\le(1+\varepsilon)\chimax$.
\end{theorem}

The proof of Theorem~\ref{thm:montgomery2025} is quite involved and also nonconstructive,%
\footnote{One key step in the proof involves finding roughly $\sqrt{n}$ rainbow
independent sets having certain lexicographic-optimality properties. This seems quite
challenging to obtain via an efficient algorithm.} 
which leads to the second question that we explore.

\medskip

\qquad 
\begin{minipage}{0.9\linewidth}
\begin{openquestion} \label{openquestion2} \label{qn2}
Is there a polynomial-time algorithm that obtains similar asymptotic guarantees for Rota's
Basis Conjecture as Theorem~\ref{thm:montgomery2025}?
\end{openquestion}
\end{minipage}

\end{enumerate}

\subsection{Our Contributions} \label{contrib}

We make substantial progress on the above questions, advancing the constructive frontier in the area of matroid-intersection
coloring. We fully resolve Question~\ref{qn2}, and resolve Question~\ref{qn1} for the case
of two matroids, thus essentially eliminating the constructive-nonconstructive gap
in the setting of two matroids. %

\medskip
Theorem~\ref{matcolthm} states our results for $k$-matroid-intersection
coloring, which addresses Question~\ref{qn1}.

\begin{theorem}[{\bf Approximation results for matroid-intersection coloring}] 
\label{matcolthm}
\
\begin{enumerate}[label=(\alph*), topsep=0.2ex, ref={\thetheorem\,(\alph*)}, itemsep=0.1ex, leftmargin=*]
\item There is a polytime algorithm that returns a $2\chimax$-coloring for
$2$-matroid-intersection coloring.
\label{thm:constructive2}\label{twomatcol}

\item There is a polytime algorithm that returns a $k(k-1)\chimax$-coloring for
$k$-matroid-intersection coloring for $k \geq 2$.
\label{thm:constructive1}\label{kmatcol}
\end{enumerate}
\end{theorem}

Theorem~\ref{thm:constructive2}, which is an immediate corollary of
Theorem~\ref{thm:constructive1}, resolves Question~\ref{qn1} for 
the case of two matroids. 
For $k>2$, while we do not match the guarantee in Theorem~\ref{thm:nonconstructive1},
the $k(k-1)\chimax$ bound in Theorem~\ref{kmatcol} significantly improves upon the
previously-known approximation ratio of $O(k\log n)$ 
(via a reduction to set cover) for $k$-matroid-intersection coloring. In particular, this
yields the first $O(1)$-approximation algorithm 
for any fixed $k$.

In contrast to prior combinatorial approaches, which (as noted earlier) do not seem
amenable to handling general matroids, our results are obtained by moving to an
LP-relaxation of a matroid-intersection view of the problem, and leveraging an LP-rounding
algorithm for matroid intersection (with multiple matroids). The rounding algorithm does
not quite yield a coloring, 
but we identify suitable structure in the output of the rounding algorithm that enables us
to convert the ``approximate'' coloring returned by the algorithm into a valid
coloring. We elaborate further on these ideas %
in Section~\ref{techoverview-kmat}.

\medskip

Our second main result is a %
constructive version of Theorem~\ref{thm:montgomery2025}, thus
fully resolving Question~\ref{openquestion2}. 
Recall that in Rota's Basis Conjecture (Conjecture~\ref{rota}), we have a matroid
$\mat_1=(\gset,\inds_1)$ and $r$ disjoint bases $B^1,\ldots,B^r$, where
$r=\rk_{\mat_1}(\gset)$, and to cast this in the setup of matroid-intersection coloring,
we take $\mat_2$ to be the partition matroid encoding that at most one element is selected
from each $B^j$ basis. Let $n=|\gset|$.

\begin{theorem} \label{thm:rotaresult} 
For any $\varepsilon>0$ and any instance of Rota's Basis Conjecture where $\chi(\mat_1)$
(or equivalently the rank $r$) is sufficiently large relative to $\varepsilon$, 
there is a randomized algorithm with running time 
$\poly\bigl(n,\frac{1}{\varepsilon}\bigr)$ that returns a
$(1+\ve)\chimax$-coloring of $\mat_1\cap\mat_2$ with high probability.%
\footnote{It suffices to ensure $\frac{1}{\poly(n)}$ success probability, 
since we can then boost the success probability to $1-\frac{1}{n^\alpha}$ for
any $\alpha>0$ by running the algorithm $\poly(n)\cdot O(\log n)$ times.}
\end{theorem}

Theorem~\ref{thm:rotaresult} actually follows from a much-more general result that we
obtain. We develop a {\em fully polynomial randomized approximation scheme} (FPRAS) for
two-matroid-intersection coloring with {\em any} two matroids (as opposed to $\mat_2$
being a partition matroid), when $\chimax$ is sufficiently large. 

\begin{theorem}\label{thm:largec_clean} 
For any $\ve>0$ and any instance $\mat_1,\mat_2$ of
two-matroid-intersection coloring with $\chimax\geq\frac{C\log n}{\ve^5}$, where $C$ is
an absolute constant, 
there is a randomized algorithm with running time $\poly\bigl(n,\frac{1}{\ve}\bigr)$ 
that returns a $(1+\ve)\chimax$-coloring with high probability. 
\end{theorem}

\begin{corollary}
Taking $\ve=\bigl(\frac{C\log n}{\chimax}\bigr)^{1/5}$ in~\Cref{thm:largec_clean}, we obtain a coloring using $\chimax+O\bigl(\log^{1/5} n\bigr)\cdot\chimax^{4/5}$ colors with high
probability; this translates to $\bigl(1+o(1)\bigr)\chimax$ colors when $\chimax=\w(\log n)$.
\end{corollary}

Theorem \ref{thm:rotaresult} is an immediate corollary of Theorem \ref{thm:largec_clean} since
in Rota's Basis Conjecture, we have $\chimax=\sqrt{n}$. 
Theorem \ref{thm:largec_clean} improves upon Theorem~\ref{thm:montgomery2025} in various ways.
First, it yields an {\em efficient} algorithm; 
second, it generalizes
Theorem~\ref{thm:montgomery2025} by virtue of handling {\em arbitrary} instances of
two-matroid-intersection coloring;  
third, the asymptotic dependence on $\chimax$, i.e., the $o(1)$ term in
$\bigl(1+o(1)\bigr)\chimax$, is explicitly spelled out;  
finally, our proof is substantially simpler than the proof of
Theorem~\ref{thm:montgomery2025}. 
As we discuss in Section~\ref{techoverview-rota}, our algorithm and its analysis are based 
on a surprisingly clean approach that exploits polyhedral results on matroid intersection,
in particular, the powerful {\em swap-rounding} technique in~\cite{ChekuriVZ11}, to
iteratively ``peel'' out common independent sets. %
In essence, we gain a striking amount of leverage from the use of this polyhedral rounding 
technique, and this is directly responsible for %
the simplicity of our proof. %

\begingroup\tolerance=1500\emergencystretch=4em
One interesting insight to emerge from our work (specifically
Theorem~\ref{thm:largec_clean}) is that asymptotically-optimal coloring is possible in
the setting of Rota's Basis Conjecture {\em not because one of the matroids is a partition
matroid, but because $\chimax$ is sufficiently large}. 
This is again a direct consequence of the polyhedral approach we utilize, which is 
essentially opaque to the structure of the underlying matroids. \par\endgroup

\medskip
It is worth pointing out that while the techniques used to obtain our two main
results (Theorem~\ref{matcolthm}, and Theorem~\ref{thm:largec_clean})
are quite different, (as noted above) a common thread underlying our results and
techniques,   
which distinguishes our work from prior approaches, is that 
we leverage polyhedral techniques, and this polyhedral viewpoint paves the way for 
obtaining our results. 

\section{Technical Overview} \label{techoverview}

At a high level, the main challenge in proving 
the existence of
good colorings for matroid-intersection coloring, is that whether two elements of the ground set can be colored
the same color is quite context-dependent, and depends on which other elements share that color. Taking the covering perspective, it is not clear how to compute common independent sets that reduce the chromatic number for the remaining uncovered elements. %
 The proof of \Cref{matcolthm} surmounts the coloring challenge, and the proof of \Cref{thm:largec_clean} surmounts the covering challenge. %
We proceed to highlight the main insights and ideas in these proofs. 

\subsection{Coloring \boldmath $k$ Matroids: Theorem~\ref{matcolthm} (see Section~\ref{iterrnd})} \label{techoverview-kmat}

Recall that prior algorithmic results focus on
the setting where all, or almost all,  matroids~\cite{partredn1,partredn2,partredn3,arndt2025} are partition matroids or are reducible to partition matroids. These techniques exploit the
fact that, in partition matroids, it is substantially easier (compared to general matroids)
to understand (and hence control) when a set is independent. 
This partition-reduction approach can handle certain combinatorial classes of matroids, but there are known classes where reduction to a partition matroid drastically increases the chromatic number~\cite{imposs_1,imposs_2}.

Our approach is different, yet fairly natural. Let $q=\chimax$. Our algorithm consists of three main steps. First, we utilize a standard observation to cast the problem as an instance of $(k+1)$-matroid intersection. Second, we utilize an LP-rounding algorithm for this matroid-intersection problem. This algorithm does not, however, produce a feasible solution. While its output can be viewed as a kind of ``approximate coloring'', this by itself is too weak a guarantee for our purposes. The third key step, which is our main conceptual contribution, is that we identify suitable structure in the output of the LP-rounding algorithm, and show that  one can exploit this structure to set up a {\em graph-coloring} problem to resolve conflicts in the approximate coloring and obtain a valid coloring.

We now elaborate on these steps.
The reduction to matroid intersection involves moving to a new ground set $\gset'$ consisting of $q$ disjoint copies $\gset_1,\ldots,\gset_q$ of the original ground set $\gset$, and considering the disjoint union of $q$ copies of each input matroid $\mat_i$, which we denote $\mat'_i$. We introduce an additional matroid $\mat'_0$ which encodes for each element in $\gset$, that at most one copy of each element $e\in U$ is picked.
It is easy to see now that a valid $q$-coloring maps to a subset $R\sse\gset'$ which is a basis of $\mat'_0$ and independent in $\mat'_i$ for each $i$, and vice versa.

We consider the natural LP-relaxation of the problem %
and utilize the LP-rounding algorithm of~\cite{LinharesOSZ20}.
Their algorithm, which we refer to as the \losz-algorithm, %
outputs a basis $R$ of $\mat'_0$ 
such that each ``color class'' $R\cap\gset_c$ is $k$-colorable in
$\mat_i$ for all $i\in[k]$, $c\in[q]$. 
However, this guarantee alone is {too coarse} for our purposes: the issue is that being $k$-colorable in each individual matroid does not mean that we can efficiently obtain a good coloring in their intersection. In general, all
one can glean from this guarantee, by way of efficiently producing a coloring, is
a $(k^k\cdot q)$-coloring %
by partitioning
$R\cap\gset_c$ into $k$ independent sets of $\mat_i$, for all $i\in[k]$ and taking all the resulting intersections.%
\footnote{As noted earlier, while it is true here that $R\cap\gset_c$ 
can be covered using $O\left(k^2\right)$ common independent sets (Theorem~\ref{thm:nonconstructive1}),
the proof of this is quite nonconstructive and does not lead to an efficient algorithm.}
Thus, this only yields a $k^k$-approximation algorithm; while this is a constant for fixed $k$, it is quite far from the $k(k-1)\chimax$-coloring that we are aiming for. 
Moreover, for $k=2$, there are instances where 
$\chi(\mat_1)=\chi(\mat_2)=2$, but $\chi(\mat_1\cap\mat_2)>2$,
so $2$-colorability in each individual matroid would not by itself yield the $2\chimax$-coloring for $2$-matroid-intersection coloring (Theorem~\ref{twomatcol}). 
The upshot is that the guarantee of~\cite{LinharesOSZ20} by itself is insufficient for
obtaining our approximation results. 

We make progress via two chief insights.
Our first insight is that 
one can identify much more structure in the set output by the \losz algorithm,
encapsulated by a property that we call {\em flexible decomposition} (see
Definition~\ref{flexdecomp}). 
We say that a set $S$ admits a $p$-flexible decomposition in a matroid $\mat$ if there exist matroids $N_1=(T_1,\inds_1), N_2=(T_2,\inds_2),\ldots, N_\ell=(T_\ell,\inds_\ell)$ such that: $T_1,\ldots,T_\ell$ is a partition of $S$,
the rank of $N_j$ is at least $|T_j| - (p-1)$ for all $j \in [\ell]$, and ($\ast$) 
the union of {\em any} collection of $N_j$-independent sets (one per index $j\in[\ell]$) yields an $M$-independent set.  
Observe that taking $p=1$ corresponds precisely to saying that $S$
is independent in $\mat$; %
on the other hand, one can also infer that if $S$ has a $p$-flexible decomposition in
$\mat$, then $S$ is $p$-colorable in $\mat$. Thus, $p$-flexible decomposition relaxes
independence, but is a {\em strengthening} of $p$-colorability.

We prove that the set $R$ output by the \losz algorithm is such that $R\cap\gset_c$ admits a 
{\em $k$-flexible decomposition in $\mat_i$} for all $i\in[k]$, $c\in[q]$, thereby
strengthening the guarantee of~\cite{LinharesOSZ20}. %
We refer to this compactly by saying that $R$ is a {\em $(k,q)$-pseudocoloring}\,%
\footnote{Note that a $(1,q)$-pseudocoloring is a $q$-coloring, so a
$(p,q)$-pseudocoloring is a relaxation of $q$-coloring (as the name suggests).} 
(see Definition~\ref{pseudocol} and Theorem~\ref{flexdecompthm}). 

Our second key insight is that the stronger guarantee provided by a
$(k,q)$-pseudocoloring can be leveraged to convert this pseudocoloring into a
$k(k-1)q$-coloring of $\mat_1,\ldots,\mat_k$ (Theorem~\ref{convthm}). We argue that, for
each color class $c\in[q]$, one can create a suitable {\em conflict graph} $G_c$ with
vertex set $R\cap\gset_c$ such that: (a) $G_c$ admits an efficiently-computable $k(k-1)$
(vertex) coloring; and (b) such a coloring translates to a cover %
of $R\cap\gset_c$ by
$k(k-1)$ common independent sets. Thus, we obtain a $k(k-1)q$-coloring. 

Roughly speaking, for a color class $c\in[q]$ and matroid $\mat_i$ where $i\in[k]$, in
order to ensure independence in $\mat_i$, due to property ($\ast$) 
it suffices to focus on a single part $T_j$ of the $k$-flexible decomposition of
$R\cap\gset_c$ in $\mat_i$ and cover $T_j$ by $N_j$-independent sets. 
Using matroid-exchange properties, this can be captured by 
defining suitable conflict pairs, encoded by edges involving elements in $T_j$. 
The union of these (vertex-disjoint) graphs for all parts $T_j$ yields a graph $H_i$ 
with maximum degree at most $k-1$, 
and since we need to ensure independence in all $k$ matroids, we take the union of $H_i$ over all
$i\in[k]$ to obtain the conflict graph $G_c$. 
Thus, $G_c$ has maximum degree at most $k(k-1)$, and a simple greedy algorithm 
yields a $k(k-1)+1$ coloring of $G_c$. We improve this using (an algorithmic proof of)
Brooks' theorem (Theorem~\ref{thm:brooks-thm}) to a $k(k-1)$-coloring by proving that $G_c$
does not contain a clique on $k(k-1)+1$ nodes. (For $k=2$, there is a simpler argument
showing that $G_c$ is bipartite, and hence $2$-colorable.)  

\smallskip
It is %
instructive to compare, 
at a high level, our approach, with the
partition-reduction approaches in prior work which achieve good approximations for combinatorial matroids.
These approaches 
reduce to partition matroids at the very outset, yielding a simple
sufficient condition for independence and allowing the
resulting instance to be handled by combinatorial means.
In contrast, our approach for general matroids leverages a 
reduction to matroid intersection and the LP-relaxation for the latter problem. 
The flexible-decomposition structure that we identify results from the output of the LP rounding algorithm and is not at all apparent at 
the beginning. %
Thus, we rely heavily on, and are guided by, the LP and the LP-rounding algorithm to isolate the desired
structure, which we finally capitalize on to ``clean up'' the pseudocoloring and obtain a genuine
coloring.
We note that the use of graph coloring in the final clean-up step, that is, encoding independence via conflict-pair avoidance, is also a component of the partition-reduction approach.

\subsection{Coloring Two Matroids when \boldmath$\chimax$ is Large: \Cref{thm:largec_clean} (see Section~\ref{sec:FPRAS})} \label{techoverview-rota}
Here we give a technical overview of the proofs underlying the constructive result in \Cref{thm:largec_clean}.

From Edmonds' Condition (see Fact \ref{onematcol}), we know that $\chimax$ is the ceiling of the
maximum over all subsets $S$ of $U$ and over all matroids $M_i$, of the density of $S$ in $M_i$,
which is the ratio $|S|/r_i(S)$ of its cardinality to its rank in $M_i$. 
So, a natural  approach to show the existence of a good coloring 
is to iteratively decide on some subsets of elements to be color classes, in such a way as to reduce the
densities of the high density subsets of the remaining elements. 
This is the approach taken in the nonconstructive proof of Theorem \ref{thm:montgomery2025} in \cite{montgomery2025}. 
However, the authors note that keeping track of these densities is very difficult, and this is one of the main
reasons that their proof is so complex. 
We avoid this combinatorial difficulty using a polyhedral approach combined with randomization.

It is helpful to first consider a  simple, but flawed, randomized strategy. Suppose we sample each element independently and uniformly at random with probability $p $. It is not too difficult to show, via standard concentration bounds and a union bound, that if we delete these sampled elements, the density of \emph{every} high density set (those with density close to $\chimax$) decreases to $(1-p+o(1))\chimax  $ with high probability if $\chimax = \omega(\log n)$.  In particular, if we choose $p= N/\chimax$ for some sufficiently large $N$, the maximum density drops by approximately $N$. The core insight here is that we do not need to identify or specifically target high density sets $S$. Because this ``blind'' uniform deletion effectively hits all sets with high concentration, it universally reduces the maximum density. This would be a perfect strategy if the randomly sampled set using $p = N / \chimax$ were a collection of $N$ common independent sets in the matroid intersection. Indeed, we could simply color the sampled elements with $N$ colors, and recurse on a subproblem with a chromatic number that is $N$ smaller, implying that only $\chimax + o\left(\chimax\right)$ iterations suffice. Of course, this method of independent sampling makes no structural progress because the resulting sampled set need not be a collection of independent sets in either matroid, let alone their intersection. However, this uniform hitting property suggests the utility of random sampling. So, we need a random sampling method to generate common independent sets that has (nearly) uniform marginals and the
 strong concentration of independent sampling. 

Fortunately, the swap rounding technique of Chekuri, Vondr\'ak, and Zenklusen~\cite{ChekuriVZ11} gives us exactly this type of guarantee. 
Swap rounding rounds a fractional point $x$ in the matroid intersection polytope $P$ to an integer point $x_R$ corresponding to a common independent set $R$. Further, $\mathbb{E}[x_R] \approx x$, and \textit{every} linear function on the coordinates of $x_R$ strongly concentrates around its expectation. Our algorithm will take $x$ to be the uniform fractional point in which each component is $1/\chimax$. It is easy to see that this point is in the matroid intersection polytope $P$ (see Claim \ref{matpoly}). It might appear that at first glance this point $x$ is combinatorially ``featureless'', revealing no information about the structure of a feasible integer coloring. 
But when swap rounding is applied $N$ times to the uniform fractional point $x$, we obtain our desired properties, 
namely that every element is included with marginal probability approximately $N/\chimax$, \emph{and}  concentration bounds on every subset of elements comparable to independent random sampling. 

Given the complexity of the nonconstructive proof of Theorem \ref{thm:montgomery2025}, our   $(1+\eps)$-approximation algorithm is surprisingly clean. For $N = \lceil \eps \chimax \rceil$, we just apply swap rounding $N$ times to the uniform fractional point $x$, in which each component is $1/\chimax$,
to obtain common independent sets $R_1, R_2, \dots, R_N$. The elements in each $R_i$ are then colored the same color. We then ``peel'' these colored elements away, re-evaluate the maximum density of the remaining uncolored elements, and recurse on the still uncolored elements. Because swap rounding guarantees tight concentration, we are able to argue that the maximum density drops reliably by roughly a $1-\eps$ factor in each step.

\section{Preliminaries and notation} \label{prelim}
For an integer $k\geq 1$, we use $[k]$ to denote the set $\{1,2,\ldots,k\}$.
We use $\R_+$ to denote the nonnegative reals; similarly $\Z_+$ denotes the nonnegative
integers. 
Recall that in the {\em $k$-matroid-intersection coloring} problem, we
are given $k$ matroids $\mat_1=(\gset,\inds_1),\ldots,\mat_k=(\gset,\inds_k)$ over a
common ground set $\gset$, and we seek to cover $\gset$ using the fewest number of common
independent sets of these $k$ matroids; a covering using $q$ common independent sets
is called a {\em $q$-coloring} of $\mat_1,\ldots,\mat_k$. 

More generally, given an {\em independence system} $N=(\gset,\inds)$ (i.e., $\inds$ is a 
downwards-closed collection of sets), and a set $S\sse\gset$, we say that $S$ is
$q$-colorable in $(\gset,\inds)$, or admits a $q$-coloring in $(\gset,\inds)$, if
$S$ can be covered by $q$ sets in $\inds$. Thus, in the $k$-matroid-intersection coloring
problem, we seek a $q$-coloring of $\gset$ in the independence system
$\Mint=\bigl(\gset,\bigcap_{i\in[k]}\inds_i\bigr)$ for the smallest value of $q$.
The chromatic number, or coloring number, of $N$, denoted
$\chi(N)$, is the smallest $q$ for which there is a $q$-coloring of $\gset$ in $N$.

We assume throughout that every matroid $\mat$ we work with is loopless, that is, $\{e\}$
is independent in the matroid, for every element $e$ in its ground set. 
(Note that otherwise, no coloring of $\mat$ exists.) 

\subsubsection*{Matroid preliminaries} We collect a few relevant facts about matroids.
Let $\mat=(\gset,\inds)$ be a matroid with rank function $\rk:2^{\gset}\mapsto\Z_+$. Two
polytopes commonly associated with $\mat$ are:
\begin{enumerate}[label=(\alph*), topsep=0.1ex, noitemsep, leftmargin=*]
\item its {\em matroid polytope}, $\Pc({\mat})$, which is the convex hull of indicator
vectors of independent sets, which admits the inequality description
$\Pc({\mat})=\bigl\{x\in\R_+^\gset:\ x(S)\leq\rk(S)\ \ \forall S\sse\gset\}$; and 
\item its {\em matroid-base polytope}, $\baseP({\mat})$, which is the convex hull of
indicator vectors of bases of $\mat$, and is given by
$\bigl\{x\in\Pc({\mat}):\ x(\gset)=\rk(\gset)\bigr\}$. 
\end{enumerate}

\begin{fact}[Edmonds' Condition ~\cite{edmonds1968matroid,schrijver_book}] \label{onematcol}
    Let $M=(\gset,\inds)$ be a matroid with rank function $r: 2^U \to \mathbb{Z}_{+}$. Then $\chi(M) = \max_{S\sse\gset}\ceil{\frac{|S|}{r(S)}}$.
\end{fact}

\begin{claim} \label{matpoly}
For any $q\geq\chi(\mat)$, we have $x=\bon/q\in\Pc({\mat})$, where $\bon$ is the all $1$'s
vector in $\R^{\gset}$.  
\end{claim}

\begin{proof}
By Fact~\ref{onematcol}, $q\geq\chi(\mat)$ implies that $|S|\leq q\cdot\rk(S)$ for all $S\sse\gset$,
i.e., $\bon/q\in\Pc(\mat)$.
\end{proof}

Let $S\sse\gset$. The {\em refinement of $\mat$ along $S$} yields the matroids
$\mat_1=\mat\vert_S=(S, \inds\cap 2^S)$, which is the restriction of $\mat$ to $S$, and 
$\mat_2=\mat/S$, which is the matroid obtained from $\mat$ by contracting $S$. We have 
$\mat_2=\bigl(\gset-S, \{A\sse\gset-S: A\cup J\in\inds\}\bigr)$, where $J$ is a maximal
independent set contained in $S$; it is well known that the definition of $\mat/S$ does
not depend on the choice of which maximal independent subset of $S$ is chosen as $J$. Further, the rank function of the matroid contraction $\mat/S$ is given by $r_{\mat/S}(A) = r(A \cup S) - r(S)$ for all $A \subseteq U-S$.

\section{\boldmath $k$-Matroid-Intersection Coloring: Proof of Theorem~\ref{matcolthm}} 
\label{iterrnd}

We now present the $k(k-1)$-approximation algorithm for $k$-matroid-intersection coloring, establishing Theorem~\ref{matcolthm}. Throughout this section, we assume $k \geq 2$. Recall that $\mat_i=(\gset,\inds_i)$, $i\in[k]$, are the $k$ input matroids, and $\chimax:=\max_{i\in[k]}\chi(\mat_i)$.
Our approximation guarantees are relative to $\chimax$, and 
also yield improved integrality-gap bounds for the natural set-covering LP-formulation of the problem.

As outlined in Section~\ref{techoverview-kmat}, we leverage the reduction from $q$-coloring to a $(k+1)$-matroid-intersection problem. Specifically, we consider a ground set $\gset'$ consisting of $q$ disjoint copies of $\gset$. We define $\mat'_0$ as the partition matroid on $\gset'$ that ensures at most one copy of each original element $e \in \gset$ is selected, and for $i \in [k]$, we let $\mat'_i$ be the disjoint union of $q$ copies of the original matroid $\mat_i$. As established previously, finding a $q$-coloring of $\mat_1, \dots, \mat_k$ is equivalent to finding a basis $R$ of $\mat'_0$ that is independent in $\mat'_1, \dots, \mat'_k$.
We take $q=\chimax$, %
which ensures that the LP-relaxation of this $(k+1)$-matroid-intersection problem is
feasible (by Claim~\ref{matpoly}).

We apply the iterative refinement algorithm of~\cite{LinharesOSZ20} (the \losz-algorithm) to this instance. The set $R$ output by the \losz algorithm does not yield a coloring. While the analysis in~\cite{LinharesOSZ20} ensures that each resulting ``color class'' $R \cap \gset_c$ is $k$-colorable in each $\mat_i$ individually, as discussed earlier, this is insufficient for obtaining even a $\poly(k)$-approximation factor for $k$-matroid-intersection coloring. We instead prove that the \losz-algorithm has much more structure than previously known. We call this structure \emph{flexible decomposition}, and exploit this structure to resolve conflicts in each color class via graph coloring.

\begin{definition}[{\bf Flexible decomposition}] \label{flexdecomp}
Let $\mat = (\gset,\inds)$ be a matroid and $S \subseteq \gset$. 
Let $p\geq 1$ be an integer. 
We say that $S$ admits a {\em $p$-flexible decomposition} in $\mat$ into matroids 
$N_1=(T_1,\inds_1), N_2=(T_2,\inds_2),\ldots,N_\ell=(T_\ell,\inds_\ell)$, or
that $N_1,\ldots,N_\ell$
is a $p$-flexible decomposition of $S$ in $\mat$, if: \linebreak 
(a) $T_1,\ldots,T_\ell$ is a partition of $S$; \ \
(b) $\rk_j(T_j)\geq |T_j|-p+1$ for all $j \in [\ell]$, where $\rk_j$ is the rank function
of $N_j$; and \ \
(c) for any collection of $N_j$-independent sets $I_j\in\inds_j$ for $j\in[\ell]$, we have
that $\bigcup_{j\in[\ell]}I_j\in\inds$.

We will also require that the matroids $N_1,\ldots,N_\ell$ are loopless, and unless otherwise stated, can be
{\em efficiently} computed, i.e., the partition $T_1,\ldots,T_\ell$ can be efficiently
computed, and the matroids $N_1, \ldots, N_{\ell}$ have efficient independence oracles. 
We will often say ``$S$ admits a $p$-flexible 
decomposition in $\mat$'' to refer to the fact there is an underlying (efficiently-computable)
collection of matroids $N_1,\ldots,N_\ell$ satisfying properties (a)--(c).
\end{definition}

Note that $S\in\inds$ iff it admits a $1$-flexible decomposition: 
if $S\in\inds$, then the matroid $M\vert_S$ yields a trivial $1$-flexible decomposition; 
conversely, 
if $N_1,\ldots,N_\ell$ is a $1$-flexible decomposition of $S$, then properties (b),
(c) imply that $T_j\in\inds_j$ for all $j\in[\ell]$ and hence, 
$\bigcup_{j\in[\ell]} T_j=S\in\inds$. 
Note also that if $S$ has a $p$-flexible decomposition in $\mat$, then  
$S$ is $p$-colorable in $\mat$: 
the combination of the $\rk_j(T_j)$-size independent sets in the $T_j$ parts, along with any
combination of (the at most $p-1$) elements excluded from these independent sets where we
take (at most) one element from each part, yields a $p$-coloring in $\mat$. 
Thus, $p$-flexible decomposition strengthens the notion of $p$-colorability.

In Section~\ref{flexdecompalg}, we show that the set $R$ output by the \losz algorithm is such that $R\cap\gset_c$ admits a $k$-flexible decomposition in $\mat_i$ for all $i\in[k]$, $c\in[q]$ (Theorem~\ref{flexdecompthm}), strengthening the guarantee in~\cite{LinharesOSZ20}. 
We call this notion of ``approximate coloring'' based on flexible decomposition a {\em pseudocoloring}.

\begin{definition}[{\bf Pseudocoloring}] \label{pseudocol}
Let $\mat_i=(\gset,\inds_i)$, $i=1,\ldots,k$ be $k$ matroids. Let $p,q\geq 1$ be
integers. A {\em $(p,q)$-pseudocoloring} of $\mat_1,\ldots,\mat_k$ is a partition
$R_1,\ldots,R_q$ of $\gset$ such that $R_c$ admits a $p$-flexible decomposition in
$\mat_i$, for all $i\in[k]$, $c\in[q]$. 
We also sometimes say that $R_1,\ldots,R_q$ is a $(p,q)$-pseudocoloring of
$\mat_1,\ldots,\mat_k$.  
\end{definition}

Thus, the \losz algorithm returns a $(k,\chimax)$-pseudocoloring. 
(Note that a $(1,q)$-pseudocoloring of $\mat_1,\ldots,\mat_k$ corresponds precisely to a
$q$-coloring, since $S\sse\gset$ admits a $1$-flexible decomposition in a matroid
$\mat=(\gset,\inds)$ iff it is independent in $\mat$; therefore, %
a pseudocoloring is a relaxation of coloring.)
We next show in Section~\ref{pseudoconv} that one can exploit this pseudocoloring and turn it
into a valid coloring incurring a $k(k-1)$ multiplicative-factor blow-up in the number of colors.
We %
solve a {\em graph coloring problem} to achieve this end.
For each color class $c\in[q]$, we create a suitable conflict graph
$G_c$ with vertex set $R\cap\gset_c$ and argue that: (a) we can efficiently compute a 
$k(k-1)$ (vertex) coloring of $G_c$; 
and (b) such a coloring translates to a cover of $R\cap\gset_c$ by $k(k-1)$ common
independent sets.  
Thus, since we start off with $q=\chimax$ color classes, we obtain a $k(k-1)\chimax$-coloring of $\mat_1,\ldots,\mat_k$.

\subsection{Flexible Decomposition via LP-Rounding and Iterative Refinement} 
\label{flexdecompalg} 

We prove the following result. 

\begin{theorem} \label{flexdecompthm} \label{cor:interface} \label{pseudocolthm}
Given matroids $\mat_i=(\gset,\inds_i)$ for $i=1,\ldots,k$, we can efficiently find a
$(k,\chimax)$-pseudocoloring of $\mat_1,\ldots,\mat_k$.
\end{theorem}

Theorem~\ref{flexdecompthm} will follow readily from the following theorem,
which strengthens Theorem 2 in~\cite{LinharesOSZ20} and is the main technical result of
this section.

\begin{theorem}\label{thm:strong_iterative} \label{iterrefthm}
Let $\matref_0 = (\gsetref_0,\indsref_0)$, and $\matref_i = (\gsetref_i,\indsref_i)$, $i\in[k]$ be matroids,
where $\gsetref_i\subseteq\gsetref:=\gsetref_0$ for all $i\in[k]$, and $w\in\R^{\gsetref}$.
Consider the following LP. 
\begin{equation}
\begin{gathered}
\max \quad w^Tx \qquad \text{s.t.} \\
x \in\baseP({\matref_0}), \quad\ \  x\vert_{\gsetref_i} \in \Pc({\matref_i}) \quad \forall i \in[k].
\end{gathered}
\tag{$\lpmat$} \label{matlp}
\end{equation}
where $x\vert_S$ denotes the vector $(x_e)_{e\in S}$, for $S\sse\gsetref$.
Let $p_1, p_2, \dots, p_k$ be positive integers such that 
$\sum_{i \in [k] : e \in \gsetref_i}\frac{1}{p_i} \leq 1$ for all $e \in\gsetref$.

If \eqref{matlp} is feasible, then one can efficiently compute $R \subseteq \gsetref$ such that
\begin{enumerate*}[label=(\alph*)]
    \item $R$ is a basis of $\matref_0$; {\ \ }
    \item $w(R) \geq\OPT_{\text{\ref{matlp}}}$; and {\ \ }
    \item $R\cap\gsetref_i$ has a $p_i$-flexible decomposition in $\matref_i$ for all
      $i \in [k]$. 
\end{enumerate*}
\end{theorem}

We prove Theorem~\ref{iterrefthm} shortly, but first we show how this easily yields
Theorem~\ref{flexdecompthm} as a corollary.

\begin{claim} \label{decompclos}
Let $\mat=(\gset,\inds)$ be a matroid. Suppose $S\sse\gset$ has a $p$-flexible
decomposition in $\mat$, and $A\sse S$. Then, $A$ also has a $p$-flexible decomposition in
$\mat$.
\end{claim}

\begin{proof}
Let $N_1=(T_1,\inds_1),\ldots,N_\ell=(T_\ell,\inds_\ell)$ be a $p$-flexible decomposition
of $S$ in $\mat$.  
Then it is easy to see that $\bigl\{N_j\vert_{T_j\cap A}\bigr\}_{j\in[\ell]}$ (after
discarding any empty $T_j\cap A$ parts)
yields a $p$-flexible decomposition of $A$ in $\mat$.
\end{proof}

\begin{proofof}{\Cref{cor:interface}}
We reduce the coloring problem to a matroid-intersection problem as discussed at the
beginning of the section, and invoke Theorem~\ref{iterrefthm} on this input.
Let $q=\chimax$.
Let $\gset_c$ be a disjoint copy of $\gset$ for all $c\in[q]$ and
$\gsetref=\gsetref_0=\bigcup_{c\in[q]}\gset_c$. Let $\matref_0=(\gsetref,\indsref_0)$ be the
partition matroid encoding that for every $e\in\gset$, at most one copy of $e$ from
$\gset_1,\ldots,\gset_{q}$ is picked. 
For each $i\in[k]$, let $\matref_i=(\gsetref,\indsref_i)$ be the disjoint union of $q$ copies of
$\mat_i$, where the $c$-th copy resides over the ground set $\gset_c$ for all $c\in[q]$.

The weight vector $w$ will be
irrelevant for our purposes, and we can take $w$ to be arbitrary. 

Observe that $x=\bon/q$ is a feasible solution to \eqref{matlp}. 
Clearly, we have $x\in\baseP({\matref_0})$ for the partition matroid $\matref_0$.
For a matroid $\matref_i$, $i\in[k]$, since $\chi(\mat_i)\leq q$, 
we also have that $\chi(\matref_i)\leq q$, due to the definition of disjoint union. It
follows from Claim~\ref{matpoly} that $x\in\Pc({\matref_i})$.  

Let $R\sse\gsetref$ be the output of Theorem~\ref{iterrefthm} for the
matroids $\matref_0,\matref_1,\ldots,\matref_k$, taking $p_i=k$ for $i\in[k]$. 
We argue that $\{R\cap\gset_c\}_{c\in[q]}$, where we interpret each $R\cap\gset_c$ set as
a subset of $\gset$, is a $(k,\chimax)$-pseudocoloring of $\mat_1,\ldots,\mat_k$. 
Since $R$ is a basis of $\matref_0$, it is immediate that $\{R\cap\gset_c\}_{c\in[q]}$ is a
partition of $\gset$.
Fix an index $i\in[k]$.
Part (c) of Theorem~\ref{iterrefthm} yields that $R \cap U_i'$ has a $k$-flexible decomposition in
$\matref_i$. By Claim~\ref{decompclos}, this implies that $R\cap\gset_c$ has a $k$-flexible
decomposition in $\matref_i$ for all $c\in[q]$. Since $\matref_i$ restricted to $\gset_c$ is
simply a copy of $\mat_i$, 
this amounts to saying that $R\cap\gset_c$ has a $k$-flexible decomposition in
$\mat_i$ for all $c\in[q]$.
\end{proofof}

\subsubsection*{Proof of Theorem~\ref{iterrefthm}}\ \\
As mentioned earlier, the algorithm for producing $R$ is the same as the
iterative-refinement based LP-rounding algorithm in~\cite{LinharesOSZ20}, which we describe 
below for completeness. 

\renewcommand{\thealgorithm}{}
\begin{algorithm}
{\small
\caption{\losz \hfill \textnormal{// Iterative-refinement algorithm in~\cite{LinharesOSZ20}}
\label{loszalg}}

\begin{algorithmic}[1]
\State Initialize $\M=\{\matref_1,\ldots,\matref_k\}$,
$\mbase=(\gbase,\indbase)\assign\matref_0$, $p_{\matref_i}=p_i$ for all $i\in[k]$, and
$R\assign\es$. 

\State Compute an extreme-point optimal solution $x^*$ to \eqref{matlp} for the matroids
$\{\mbase\}\cup\M$. \label{lpsolve}

\State For every $e\in\gbase$ with $x^*_e=0$, delete $e$ from $\mbase$ and all matroids in
$\M$ containing $e$.
\State For every $e\in\gbase$ with $x^*_e=1$, contract $e$ in $\mbase$ and all matroids in
$\M$ containing $e$; also add $e$ to $R$. \label{contract}
\State If the ground-set of some matroid in $\mat\in\M$ becomes empty as a result of these
deletions and contractions, remove $\mat$ from $\M$. \label{mupdate}
\Comment{Note that deletions and contractions remove elements from $\gbase$.}

\State \IfThen{$\gbase=\es$} \Return $R$.

\While{there is a matroid $\mat=(\gset,\inds)\in\M$, set $S\subsetneq\gset$,
  $S\neq\es$, with $x^*(S)=\rk_{\mat}(S)$} \label{mchoose}
\Comment{$\rk_{\mat}$ is rank f'n. of $\mat$}
\State Update $\M\assign\M\sm\{\mat\}\cup\{\mat\vert_S, \mat/S\}$ and set
$p_{\mat\vert_S}=p_{\mat/S}=p_{\mat}$. \label{refine}
\EndWhile

\State Find a matroid $\mat=(\gset,\inds)\in\M$ with
$|\gset|\leq\rk_{\mat}(\gset)+p_{\mat}-1$ and remove $\mat$ from $\M$. \textbf{goto}
step~\ref{lpsolve}. \label{mdrop}
\end{algorithmic}
}
\end{algorithm}
\renewcommand{\thealgorithm}{\arabic{algorithm}}

Linhares et al.~\cite{LinharesOSZ20} show that the algorithm is well-defined, terminates
in polynomial time, and that parts (a) and (b) of the theorem statement hold. We do not repeat
their arguments here and only note that (a) and (b) follow because we only pick elements
with $x^*_e=1$, and $x^*$ yields a feasible solution to the new LP that is solved when we
go to step~\ref{lpsolve}. 
We focus on proving part (c), which strengthens the claim in~\cite{LinharesOSZ20} that
$R\cap\gsetref_i$ is $p_i$-colorable in $\matref_i$ for all $i\in[k]$.

Note that each matroid in $\M$ arises from a specific input matroid. Fix an
input matroid $\matref_i$. Note that any matroid $M$ that arises from $\matref_i$ via
refinement in step~\ref{refine} or contraction in step~\ref{contract} has $p_M=p_i$.
We can describe the sequence of matroids that arise from $\matref_i$ via a
tree. The nodes of this tree are essentially the matroids arising from $\matref_i$ that
are present in $\M$ at some point in the algorithm.
The root node is $\matref_i$. Consider a node $\mat=(\gset,\inds)$ of the tree. If $\mat$
and $S\sse\gset$ are chosen in step~\ref{mchoose} for refinement, then 
the children of $\mat$ are the matroids $\mat\vert_S$ and $\mat/S$. It will be convenient to
also include matroids obtained via contraction as part of this tree. If we contract
$S\subsetneq\gset$ in step~\ref{contract}, %
the children of $\mat$ are the free matroid $(S,2^S)$ (which is the same as
$\mat\vert_S$, since $S\in\inds$) and $\mat/S$. 
(In essence, contracting $S$ can also be viewed as refining $\mat$ along $S$.) 
The free matroid on $S$ is not part of $\M$, so does not create any children and is a
leaf of the tree; also if $S=\gset$ is contracted, then $\mat$ gets removed from $\M$ in
step~\ref{mupdate} and is a leaf of the tree. 
Observe that every leaf of this tree is a free matroid, or a matroid removed from
$\M$ in step~\ref{mdrop}. 
Clearly, this tree can be constructed efficiently as the algorithm unfolds.

We now argue that for every node $\mat=(\gset,\inds)$ of the tree, $R\cap\gset$ admits a
$p_i$-flexible decomposition in $\mat$, by induction starting from the leaves of the tree.  
Applying this to the root node shows that $R\cap\gsetref_i$ has a $p_i$-flexible
decomposition in $\matref_i$.
For the base case, when $\mat$ is a leaf, this is certainly true if $\mat$ is a free
matroid since then $R\cap\gset\in\inds$. 
Otherwise, $\mat$ must have been dropped from $\M$ in step~\ref{mdrop} and we
have $p_{\mat}=p_i$, in which case 
the matroid $\mat\vert_{R\cap\gset}$ yields a trivial $p_i$-flexible decomposition.

Next, suppose that we have a non-leaf node $\mat=(\gset,\inds)$ with children
$\bmat=\mat\vert_S$ (with ground-set $S$), $\mat''=\mat/S$ (with ground-set $\gset-S$) for
some $S\subsetneq\gset$, $S\neq\es$. Inductively, we have that $R\cap S$ admits a
$p_i$-flexible decomposition in $\bmat$ into matroids
$\bN_1=(A_1,\binds_1),\ldots,\bN_\ell=(A_\ell,\binds_\ell)$, 
and $R\cap(\gset-S)$ admits a
$p_i$-flexible decomposition in $\mat''$ into matroids $N''_1=(B_1,\inds''_1),\ldots,N''_m=(B_m,\inds''_m)$. 

We claim that
$\bN_1,\ldots,\bN_\ell,N''_1,\ldots,N''_m$ is a $p_i$-flexible decomposition of $R\cap\gset$ in
$\mat$.  
This follows more or less immediately from the definition of $p_i$-flexible
decomposition. It is clear that $A_1,\ldots,A_\ell,B_1,\ldots,B_m$ partition
$R\cap\gset$. It is also immediate that $\rk_{\bN_j}(A_j) \geq |A_j|-p_i+1$ for all $j \in
[\ell]$, and $\rk_{N''_j}(B_j) \geq |B_j|-p_i+1$ for all $j \in [m]$, as this is
guaranteed by the $p_i$-flexible decompositions of $\bmat$ and $\mat''$. 

Finally, let $\bI_1,\ldots,\bI_\ell$ be such that $\bI_j$ is $\bN_j$-independent for all
$j\in[\ell]$, and $I''_1,\ldots,I''_m$ be such that $I''_j$ is $N''_j$-independent for all $j\in[m]$. 
Then, by property (c) of $p_i$-flexible decomposition,
$\bI=\bigcup_{j\in[\ell]}\bI_j$ is independent in $\bmat=\mat\vert_S$, and
$I''=\bigcup_{j\in[m]}I''_j$ is independent in $\mat''=\mat/S$. It follows
that $\bI\cup I''\in\inds$. Thus, $\bN_1,\ldots,\bN_\ell,N''_1,\ldots,N''_m$ satisfy
property (c), completing the induction step.

The proof above also shows that the flexible decomposition can be efficiently
constructed. Indeed, the $p_i$-flexible decomposition of $R\cap\gsetref_i$ is simply
given by the leaf matroids $\mat=(\gset,\inds)$ of
the tree constructed for $\matref_i$, restricted to $R\cap\gsetref_i$.
\hfill \qed

\subsection{Converting a Pseudocoloring into a Valid Coloring}
\label{pseudoconv} 
We now show how to convert the pseudocoloring given by Theorem~\ref{flexdecompthm} to a
valid coloring.

\begin{theorem}\label{thm:main_approx} \label{convthm}
Given a $(k,q)$-pseudocoloring $R_1,\ldots,R_q$ of
$\mat_1=(\gset,\inds_1),\ldots,\mat_k=(\gset,\inds_k)$, one can efficiently obtain a
$k(k-1)q$-coloring of $\mat_1,\ldots,\mat_k$.
\end{theorem}

Combining this with Theorem~\ref{flexdecompthm}, we immediately obtain a
$k(k-1)\chimax$-coloring of $\mat_1,\ldots,\mat_k$ in polynomial time, thereby proving
Theorems~\ref{twomatcol} and~\ref{kmatcol}. 
The remainder of this section is devoted to the proof of Theorem~\ref{convthm}.

Recall that $R_1,\ldots,R_q$ being a $(k,q)$-pseudocoloring means that $R_c$ has a
polynomial-time-computable $k$-flexible decomposition %
in $\mat_i$ for all $i\in[k]$, $c\in[q]$. %
For each $c\in[q]$, 
we will define a conflict graph $G_c$ with vertex set $R_c$ such that any vertex coloring
of $G_c$ using some $s$ colors will translate to a cover of $R_c$ by $s$ common
independent sets of $\mat_1,\ldots,\mat_k$. 
Moreover, we will argue that $G_c$ can be efficiently colored using at most $k(k-1)$
colors. This yields a $k(k-1)q$-coloring of $\mat_1,\ldots,\mat_k$, proving
Theorem~\ref{convthm}. 
This translation from graph coloring to matroid-intersection coloring is {\em enabled}
by the structure imposed by flexible decompositions, which we crucially exploit. 

In the sequel, we fix a color class $R_c$, where $c\in[q]$, and define $G_c$ and prove the
above two properties for $G_c$. %
Let $\bigl\{N^i_j=(T^i_j,\inds^i_j)\bigr\}_{j \in [\ell(i)]}$ be a polynomial
time-computable $k$-flexible decomposition of $R_c$ in 
$\mat_i$, for $i\in[k]$. For each matroid $i\in[k]$, define the following graph 
$H_i=(R_c,E_i)$ with vertex-set $R_c$. For each $j\in[\ell(i)]$, let 
$A^i_j\sse T^i_j$ be an independent set of $N^i_j$ with $|A^i_j|\geq|T^i_j|-k+1$. Such a set is
guaranteed by property (b) of flexible decomposition (see Definition~\ref{flexdecomp}),
and can be easily found in polynomial time. By the matroid-exchange property, for every 
$e\in T^i_j-A^i_j$, we can identify some $e'\in A^i_j$ such that
$A^i_j\cup\{e\}-\{e'\}\in\inds_j^i$; we include $(e,e')$ as an edge in $E_i$. We also
include edges between all pairs of elements in $T^i_j-A^i_j$ in $E_i$. 
We do this for all $j\in[\ell(i)]$ to obtain the graph $H_i$. Thus, $H_i$ is the union of
vertex-disjoint subgraphs, one for each part $T^i_j$ of the $k$-flexible decomposition in
$\mat_i$. 

The conflict graph $G_c$ for color class $R_c$ is the union
$\bigl(R_c,\bigcup_{i\in[k]}E_i\bigr)$ of all $H_i$'s. 
(Note that we include only one copy of an edge even if the edge is present in multiple
$H_i$ graphs.) 

Lemma~\ref{gcoltomatcol} argues that a vertex coloring of $G_c$ yields a cover of $R_c$ by common
independent sets, and Lemma~\ref{confcol} shows that $G_c$ can be colored using at most
$k(k-1)$ colors.

\begin{lemma}\label{lem:feasible} \label{gcoltomatcol}
Let $C_1, \ldots, C_s$ be a vertex coloring of $G_c$ using $s$ colors. Then, 
$C_1,\ldots,C_s$ is an $s$-coloring of $R_c$ in $\mat_1,\ldots,\mat_k$.
\end{lemma}

\begin{proof}
Consider any $h\in[s]$.
We argue that $C_h$ is a common independent set of $\mat_1,\ldots,\mat_k$, which will prove the lemma.
We have that $C_h$ is a stable set of $G_c$, i.e., there are no edges in
$G_c$ between any two nodes in $C_h$.
We claim that $C_h\cap T^i_j$ is independent in $N_j^i$ for all $i\in[k]$ and $j \in [\ell(i)]$. 
By property (c) of flexible decomposition, this implies that
$C_h=C_h\cap R_c=\bigcup_{j\in[\ell(i)]}(C_h\cap T^i_j)$ is independent in $\mat_i$, so
it suffices to prove the claim.

To see this, recall that $A^i_j\sse T^i_j$ is the independent set in $N_j^i$ that
was used to define the graph $H_i$ (and hence, $G_c$). If $C_h\cap T^i_j\sse A^i_j$, then
we are done. Otherwise, $C_h\cap T^i_j$ contains exactly one element $e\in T^i_j-A^i_j$,
since $H_i$ contains a clique on the elements in $T^i_j-A^i_j$. In this case, if
$(e,e')$ is an edge in $H_i$, where $e'\in A^i_j$, then we know that $e'\notin C_h$ and
$A^i_j\cup\{e\}-\{e'\}\in\inds_j^i$. It follows that 
$C_h\cap T^i_j\sse A^i_j\cup\{e\}-\{e'\}$ and so is independent in $N_j^i$.
\end{proof}

To show that $G_c$ can be efficiently colored using $k(k-1)$ colors, we utilize Brooks'
theorem from graph theory, which has a simple proof via a clever greedy
algorithm due to Lovasz~\cite{lovasz1975three}. 

\begin{theorem}[Brooks, 1941]\label{thm:brooks-thm}
Let $G$ be a connected graph with maximum degree $\Delta(G) > 2$. Then
$\chi(G) \leq \Delta(G)$ unless $G = K_{\Delta+1}$,
and a coloring of $G$ using at most $\Dt(G)$ colors can be computed in polynomial time. 
\end{theorem}

\begin{lemma}\label{lem:apply-brooks} \label{confcol}
The conflict graph $G_c$ can be efficiently colored using at most $k(k-1)$ colors.
\end{lemma}

\begin{proof}
When $k=2$, there is a particularly simple argument that is worth noting
separately. In this case, note that each $H_i$ graph is a {\em matching} (as
$|T^i_j-A^i_j|\leq 1$ for all $j\in[\ell(i)]$) and $G_c$ is a union
of two matchings. So every path or cycle in $G_c$ is an alternating path or alternating
cycle with respect to these matchings. Therefore, $G_c$ is bipartite, and hence
$2$-colorable. 

Now consider general $k \geq 3$. Note that each $H_i$ graph has maximum degree at most $k-1$, and
so $\Dt=\Dt(G_c)\leq k(k-1)$. If $\Dt<k(k-1)$, then a straightforward greedy algorithm
yields a $(\Dt+1)$-coloring. 
So suppose $\Dt=k(k-1)$. We argue that $G_c$ does not contain a clique $K$ on $\Dt+1$ nodes.  
Then by Brooks' Theorem (\Cref{thm:brooks-thm}), $G_c$ can be colored efficiently using
$\Delta=k(k-1)$ colors. 

For a vertex $v$ and graph (or edge-set) $F$, we use $\dt_F(v)$ to denote the edges of $F$
incident to $v$.
Suppose, for a contradiction, that $G_c$ contains a clique $K$ on $\Dt+1$ nodes. Observe that the clique $K$ is a connected component of $G_c$, because the maximum degree of $G_c$ is $\Dt$. Further, for any node $v\in K$, we have
$|\dt_{G_c}(v)|=k(k-1)\leq\sum_{i\in[k]}|\dt_{H_i}(v)|$, which implies that
$|\dt_{H_i}(v)|=k-1$ for all $i\in[k]$. This has to hold for all 
$v\in K$. Thus, $H_i$ is a $(k-1)$-regular graph on $K$ for all $i \in [k]$.

Further, recall that each $H_i$ is the union of vertex-disjoint subgraphs, one subgraph
for each part $T_j^i$ of the $k$-flexible decomposition of $R_c$ in $M_i$. Each of these
subgraphs contains a clique on $|T_j^i-A_j^i| \leq k-1$ nodes along with one extra edge
$(e, e')$ for each $e \in T_j^i - A_j^i$ to some $e' \in A_j^i$. In order for this
subgraph to be $(k-1)$-regular on $K$, it must be a $k$-clique. Thus, $H_i$ restricted to
$K$ is the union of vertex-disjoint $k$-cliques. %

But $K$ is a graph on $\Delta+1 = k(k-1)+1$ nodes, which is not an integer multiple of $k$, 
yielding a contradiction. 
\end{proof}

\begin{remarks} \ \\[-1.5ex]

\begin{enumerate}[label=\arabic*., topsep=0ex, itemsep=0.2ex, leftmargin=*]
\item If the input
is a $(p,q)$-pseudocoloring (where $p$ need not be $k$), then we can still create the
conflict graphs as above, and Lemma~\ref{gcoltomatcol} continues to hold. We have
$\Dt(G_c)\leq k(p-1)$, so a trivial greedy algorithm shows that 
$G_c$ can be colored with at most $k(p-1)+1$ colors, and we obtain a
$\bigl(k(p-1)+1\bigr)q$ coloring of $\mat_1,\ldots,\mat_k$.  

When $k=3$, our algorithm produces a $6\chimax$ matroid-intersection coloring. In this
case, each $H_i$ graph (constructed for a given $c\in[q]$) consists of a collection of
vertex-disjoint triangles or paths of length at most $3$, and so $G_c$ is a union of three
such graphs. The hardest case here seems to be when each $H_i$ consists of vertex-disjoint
triangles. It is conjectured that in this case that a graph such as $G_c$ can be
$5$-colored (see Question 1.8 in~\cite{AharoniAAHHJKN18}). If 
this holds and a $5$-coloring can be found efficiently, then this would likely yield an
improved efficiently-computable $5\chimax$-coloring for $3$ matroids. This is notable as
it would match the $5\chimax$ bound for $3$-matroid-interection coloring obtained
by~\cite{AharoniBGK25} via nonconstructive means.

\item Our $k(k-1)\chimax$-coloring result for $k$-matroid-intersection coloring
also implies (constructively) essentially a $k(k-1)$ upper bound on the integrality gap of
the natural covering-LP formulation of the problem. This is because one can argue that the
$\chimax$ is at most the ceiling of the optimal value of this LP.
\end{enumerate}
\end{remarks}

\section{Two-Matroid-Intersection Coloring: Proof of Theorem \ref{thm:largec_clean}} \label{sec:FPRAS}

In this section, we present a fully polynomial randomized approximation scheme (FPRAS) for coloring the intersection of two matroids when the maximum of their chromatic numbers is large. In particular, we prove Theorem~\ref{thm:largec_clean}. For simplicity of analysis, we actually prove the following modified version of Theorem~\ref{thm:largec_clean}, which states a seemingly weaker guarantee without the use of a constant $C > 0$ in the bound. We show in Appendix \ref{app:largec_techclean} that Theorem \ref{thm:largec_technical} readily implies Theorem \ref{thm:largec_clean}.

\begin{theorem}\label{thm:largec_technical}
For any $\ve \in (0, 0.001)$ and any instance $\mat_1,\mat_2$ of
two-matroid-intersection coloring with $\chimax\geq\frac{\log n}{\ve^5}$,
there is a randomized algorithm with running time $\poly\bigl(n,\frac{1}{\ve}\bigr)$ 
that returns a $(1+400\ve)\chimax$-coloring with probability at least $1-1/n^2$.
\end{theorem}

Our algorithm for Theorem \ref{thm:largec_technical} will in fact produce a \textit{covering} of $\Mint = M_1 \cap M_2$, which is synonymous with coloring, because a covering can be trivially converted into a coloring by removing duplicate elements.

In Section \ref{subsec:asymp_background}, we give background on a central component of our algorithm, the swap-rounding algorithm for matroid intersection \cite{ChekuriVZ11}. In Section \ref{subsec:asymp_alg}, we state our algorithm. In Section \ref{subsec:asymp_analysis}, we analyze our algorithm and prove Theorem \ref{thm:largec_technical}.

\subsection{Background on Swap-Rounding Algorithm \cite{ChekuriVZ11}}\label{subsec:asymp_background}
A central component to our FPRAS for coloring the intersection of two matroids when $\chimax$ is large is the swap-rounding algorithm for matroid intersection \cite{ChekuriVZ11}. This allows us to round arbitrary fractional points in the matroid intersection polytope to common independent sets with strong expectation and concentration properties. In particular, the main theorem of \cite{ChekuriVZ11} we leverage is the following:

\begin{theorem}[\!\cite{ChekuriVZ11}]\label{thm:new_swap_round}
Given matroids $M_1 = (U, \mathcal{I}_1)$ and $M_2 = (U, \mathcal{I}_2)$ on a common ground set $U$, a point $\mathbf{x} \in \mathcal{P}(M_1) \cap \mathcal{P}(M_2)$, and a constant $\gamma \in (0, 1/2]$, procedure 
SwapRound$(M_1, M_2, \mathbf{x}, \gamma)$ returns a common independent set $R \in \mathcal{I}_1 \cap \mathcal{I}_2$ with the following two properties.

\begin{enumerate}[label=$\bullet$]
    \item \textbf{Marginal Values} (Theorem 1.1 in \cite{ChekuriVZ11}): We have $\mathbb{E}[\bon_R] = (1-\gamma)\mathbf{x}$.
    \item \textbf{Concentration} (Lemma 4.5 in \cite{ChekuriVZ11}): For any constant $\lambda \in [-1/4, 1/4]$, integer $p \geq 2/\gamma$, and vector $\mathbf{a} \in [0, 1/p]^U$,

    \[ \mathbb{E}\left[e^{\lambda(a(R)-\mu_0)}\right] \leq e^{2\lambda^2 \mu_0} \]

    where $\mu_0 = (1-\gamma)\mathbf{a} \cdot \mathbf{x}$.
\end{enumerate}
\end{theorem}

Note that the ``Concentration'' property is stated in terms of the stronger moment generating function (MGF) bound, since we will need this stronger bound in our analysis.

The algorithm for Theorem \ref{thm:new_swap_round} starts by decomposing $\mathbf{x}$ into a convex combination of common independent sets in $\mathcal{I}_1 \cap \mathcal{I}_2$, which can be computed in polynomial time (see \cite{schrijver_book}). The common independent sets are sequentially merged until only one remains, which is the output set $R$. Given two common independent sets $I, J \in \mathcal{I}_1 \cap \mathcal{I}_2$, the merge operation involves constructing an exchange graph between $I, J$ in $M_1 \cap M_2$, finding a collection of paths and cycles representing feasible swaps between the two sets, and picking one at random to perform, which decreases the size of the symmetric difference $I \Delta J$. This is repeated until the two sets are the same, giving the merged set of the original $I, J$.

\subsubsection*{Runtime}
A key property of SwapRound, which is not explicitly stated in \cite{ChekuriVZ11}, is that its runtime is polynomial in $n = |U|$ \textit{and} $1/\gamma$ \cite{chekuri_personal_communication}. The fact that the runtime is also polynomial in $1/\gamma$ allows us to show our algorithm for two-matroid intersection coloring when $\chimax$ is large is an FPRAS and not just a PRAS. This distinction allows us to show our $(1+\eps)$-approximation algorithm is efficient even for $\eps$ as small as $\text{poly}(1/n)$, instead of only for constant $\eps > 0$.

\subsubsection*{Utility of SwapRound} Theorem \ref{thm:new_swap_round} implies powerful concentration properties and has found a wide variety of applications. In our setting, we apply Theorem \ref{thm:new_swap_round} in a few specific ways, namely we always use a uniform point $\alpha \bon \in \mathcal{P}(M_1) \cap \mathcal{P}(M_2)$; we only care about linear functions of the form $a_i \in \{0, 1\}$ for all $i$, which correspond to $a(R) = |R \cap S|$ for some subset of elements $S$; and we only care about negative values of $\lambda$, which induce a lower tail concentration bound. Allowing vectors $\mathbf{a} \in [0, 1]^U$ requires a basic rescaling. We now state a corollary of Theorem \ref{thm:new_swap_round} which will be all that we need for our setting.    

\begin{corollary}\label{cor:new_swap_round}
The input to SwapRound$(M_1, M_2, \alpha \bon, \gamma)$ consists of two matroids $M_1 = (U, \mathcal{I}_1)$ and $M_2 = (U, \mathcal{I}_2)$ on common ground set $U$, a constant $\alpha \in [0, 1]$ such that $\alpha \bon \in \mathcal{P}(M_1) \cap \mathcal{P}(M_2)$, and a constant $\gamma \in (0, 1/2]$. The output is a common independent set $R \in \mathcal{I}_1 \cap \mathcal{I}_2$ with the following two properties.

\begin{itemize}
    \item \textbf{Marginal Values:} For all elements $e \in U$, we have $\Pr[e \in R] = (1-\gamma)\alpha$.
    \item \textbf{Concentration:} For any constant $\lambda \in [0, \gamma/12]$ and subset of elements $S \subseteq U$,

    \[ \mathbb{E}\left[e^{-\lambda(|R \cap S|-\mu)}\right] \leq e^{6\lambda^2 \mu/\gamma} \]

    where $\mu = (1-\gamma)\alpha|S|$.
\end{itemize}
\end{corollary}

\begin{proof}
The ``Marginal Values'' condition follows immediately from \Cref{thm:new_swap_round}. For the ``Concentration'' condition, given a subset of elements $S \subseteq U$, let $p = \lceil 2/\gamma \rceil$ and $\mathbf{a} = \bon_S / p$. Note that $\mu_0 = (1-\gamma)\mathbf{a} \cdot \mathbf{x} = (1-\gamma)\alpha|S|/p = \mu/p$. Thus, \Cref{thm:new_swap_round} guarantees that

\begin{alignat}{1}
\mathbb{E}\left[e^{-\lambda(|R \cap S| - \mu)}\right] %
&= \mathbb{E}\left[e^{-p\lambda(\mathbf{a}(R) - \mu_0)}\right] \notag \\ %
&\leq e^{2(-p\lambda)^2\mu_0} 
= e^{2p\lambda^2\mu} \label{eq:swap_cor_2} \\ %
&\leq e^{6\lambda^2\mu/\gamma} \label{eq:swap_cor_4}
\end{alignat}

The inequality in (\ref{eq:swap_cor_2}) follows by \Cref{thm:new_swap_round} and the fact that $p < 2/\gamma+1 < 3/\gamma$, so $p\lambda \leq (3/\gamma)(\gamma/12) = 1/4$; %
the final inequality (\ref{eq:swap_cor_4}) follows again by $p < 3/\gamma$.
\end{proof}

Intuitively, \Cref{cor:new_swap_round} says that in the SwapRound procedure, the number of elements in the common independent set $R$ which intersect any given set $S$  has an exponential lower tail bound concentrating around the expectation.

Notice that Claim \ref{matpoly} shows that, for any matroid $M$, the vector $\bon/q\in\Pc(M)$ for any $q\geq\chi(M)$. Hence, 
if $\chimax$ is the maximum chromatic number of $M_1$ and $M_2$, then $\alpha = 1/\chimax$ can be used as a parameter in SwapRound$(M_1, M_2, \alpha \bon, \gamma)$.

\subsection{Algorithm}\label{subsec:asymp_alg}
In this section, we present our FPRAS for two-matroid intersection coloring when $\chimax$ is large. The algorithm proceeds in $\ell$ rounds. In each round, the algorithm samples several common independent sets using SwapRound. Each sampled common independent set will get its own color class in the final coloring. The elements in the sampled independent sets are removed from the ground set and the algorithm proceeds to the next round. Once the chromatic number of remaining elements becomes small enough (we choose $\ell$ so that this happens after $\ell$ rounds with high probability), the remaining elements are colored by applying the 2-approximation given by Theorem~\ref{thm:constructive2}. See~\Cref{alg:peeling}.

\begin{algorithm}[H]
{\small

\caption{FPRAS for Two-Matroid Intersection Coloring for Large $\chimax$}
\label{alg:peeling}

\begin{algorithmic}[1]
\Require Matroids $M_1 = (U, \mathcal{I}_1)$ and $M_2 = (U, \mathcal{I}_2)$, $\eps \in (0, 0.001)$
\State Initialize $\ell \gets \left\lceil\frac{\log \eps}{\log\left(1-\eps+100\eps^2\right)}\right\rceil$, $U^{(0)} \gets U$

\Statex
\For{$i = 0$ to $\ell - 1$} \Comment{Phase 1: Iterative Peeling}
    \State Let $M_1^{(i)}, M_2^{(i)}$ be $M_1, M_2$ restricted to $U^{(i)}$
    \State $\chimax^{(i)} \gets \max\left\{\chi\left(M_1^{(i)}\right), \chi\left(M_2^{(i)}\right)\right\}$
    \State $c_i \gets \lceil \eps \chimax^{(i)} \rceil$
    \State Sample $c_i$ common independent sets $\{I_1, \dots, I_{c_i}\}$ using $\text{SwapRound}\left(M_1^{(i)}, M_2^{(i)}, 1/\chimax^{(i)}, \eps\right)$
    \State $U^{(i+1)} \gets U^{(i)} \setminus \bigcup_j I_j$ \Comment{Remove all sampled elements}
\EndFor

\Statex
\State Let $\mathcal{C}_{\mathrm{final}}$ be a 2-approximate covering of $U^{(\ell)}$ (Theorem~\ref{thm:constructive2}) \Comment{Phase 2: Final Clean-up}

\Statex
\State \textbf{Output:} All sampled sets from Phase 1 and the sets in $\mathcal{C}_{\mathrm{final}}$
\end{algorithmic}
}
\end{algorithm}

\subsection{Analysis (Proof of Theorem \ref{thm:largec_technical})}\label{subsec:asymp_analysis}

In this section, we prove Theorem \ref{thm:largec_technical}. Notice that the overall cost of our algorithm (i.e. the number of color classes it produces) is the sum of the costs in Phase 1 and Phase 2. In Phase 1, we pay $c_i = \lceil \varepsilon \chimax^{(i)} \rceil$ in each iteration $i$. The cost of Phase 2 depends on the chromatic number of the leftover elements after $\ell$ such Phase 1 iterations. Our key claim is that, with high probability, the maximum chromatic number of the matroids drops by at least a roughly $1-\varepsilon$ factor in each iteration of Phase 1. This is formalized in Lemma \ref{lemma:asymp_key}, and is enough to prove \Cref{thm:largec_technical}, as we show below. 

\begin{lemma}\label{lemma:asymp_key}
If $\chimax \geq \log n/\eps^5$ and the chromatic number $\chi$ of a matroid at the start of a round is at least $\eps \chimax$, then the chromatic number of the restricted matroid at the end of the round is at most $\left(1-\eps+100\eps^2\right)\chi$ with probability at least $1-1/n^{10}$. 
\end{lemma}

We prove Lemma \ref{lemma:asymp_key} shortly, but first we show how this leads to \Cref{thm:largec_technical}.

\begin{proof}[Proof of Theorem \ref{thm:largec_technical}]
The chromatic number of both matroids after $i \leq \ell$ rounds is at most $\left(1-\eps+100\eps^2\right)^i \chimax$ with probability at least $1-2/n^{9} \geq 1-1/n^2$ by union bounding over Lemma \ref{lemma:asymp_key} and using $\ell \leq n$. Thus the total number of sampled sets is at most
\begin{equation*}
\begin{split}
\sum_{i=0}^{\ell-1} c_i & \leq \ell+\sum_{i=0}^{\ell-1} \eps \chimax^{(i)} \\
&\leq \ell + \eps \cdot \sum_{i=0}^{\ell-1} \left(1-\eps+100\eps^2\right)^i \chimax \\
&\leq \eps\chimax + \eps \cdot \sum_{i=0}^{\infty} \left(1-\eps+100\eps^2\right)^i \chimax \\
&= \eps\chimax+\frac{\eps \chimax}{1-(1-\eps+100\eps^2)} \\
&= \eps\chimax+\frac{\eps \chimax}{\eps-100\eps^2} \\
&< (1+200\eps)\chimax
\end{split}
\end{equation*}
where we use

\[ \ell = \left\lceil \frac{\log \eps}{\log(1-\eps+100\eps^2)} \right\rceil \leq \frac{2\log\eps}{-\eps} \leq \frac{2}{\eps^2} \leq \eps \chimax \]
and we use $\eps < 0.001$ in the final inequality. After $\ell$ rounds, the chromatic number of the restricted matroids is at most
\begin{equation*}
\begin{split}
\left(1-\eps+100\eps^2\right)^\ell \chimax & = \exp\left(\ell\log\left(1-\eps+100\eps^2\right)\right)\chimax \\
&\leq \exp(\log \eps)\chimax 
= \eps\chimax
\end{split}
\end{equation*}

Thus the $2$-approximate covering will use $2\eps\chimax$ sets, and so the algorithm will return a covering of $M_1 \cap M_2$ using at most $(1+200\eps)\chimax + 2\eps\chimax < (1+400\eps)\chimax$ sets with probability at least $1-1/n^2$.
\end{proof}

We now turn to proving Lemma \ref{lemma:asymp_key}. Recall that for a single matroid $\mat=(\gset,\inds)$, we have the explicit formula $\chi(\mat)=\max_{S\sse\gset}\ceil{\frac{|S|}{\rk_{\mat}(S)}}$ (Fact \ref{onematcol}). %
It is easy to see that in the  maximization expression, we can restrict to ``flats" of the matroid, where a flat is a maximal subset of a given rank. Thus, to argue that the chromatic number of each matroid drops after an iteration of sampling, we need only prove that a large fraction of \textit{all} high-density flats are covered by the sampled sets. There may be exponentially many such subsets, so care is required in order to make a union bound work. 

Hence, our next step is to show that in each iteration of Phase 1, and for every subset $S$ of the current ground set, the set of sampled elements $T$ covers (approximately) an $\eps$ fraction of $S$. Importantly, we establish an exponential tail bound on the probability that this does not occur. This is formalized in \Cref{lemma:asymp_subset}.

\begin{lemma}\label{lemma:asymp_subset}
Let $i \in [0, \ell-1]$ be a given round. Let $T_i$ be the union of the sampled common independent sets in the $i$'th round. If $\chimax \geq \log n/\eps^5$ and $\chimax^{(i)} \geq \eps\chimax$, then for all subsets $S \subseteq U^{(i)}$, and all $q\in[4,1/\eps]$, we have
\[ \Pr\left[|T_i \cap S| \leq \left(\eps-q\eps^2\right)|S|\right] \leq \exp\left(- q^2\eps^4|S|/100\right)\]
\end{lemma}

\begin{proof}
Let $S \subseteq U^{(i)}$ be arbitrary. Let $I_1, I_2, \dots, I_{c_i}$ be the sampled independent sets in the $i$'th round. Let $\Delta_j = \left|I_j \cap \left(S \setminus \cup_{h < j} I_h\right)\right|$ be the number of new elements covered by $I_j$ in $S$ for $j = 1, 2, \dots, c_i$. Note that $T_i = \cup I_j$ and $|T_i \cap S| = \sum \Delta_j$.

Define the random variable $X_j = |I_j \cap S_j|$ for each $j \in [c_i]$, where $S_j$ is a subset of $(1-\eps)|S|$ elements chosen in the following way. If the number of covered elements in $S$ is low before revealing $I_j$, specifically $|(\cup_{h < j} I_h) \cap S| < \eps|S|$, then let $S_j$ be an arbitrary collection of $(1-\eps)|S|$ uncovered elements in $S$, i.e. $S_j \subseteq S \setminus (\cup_{h < j} I_h)$. If the number of covered elements in $S$ is high before revealing $I_j$, that is, $|(\cup_{h < j} I_h) \cap S| \geq \eps|S|$, we let $S_j$ be an arbitrary collection of $(1-\eps)|S|$ elements in $U^{(i)}$.

Note that $X_j \leq \Delta_j$ if the first condition $|(\cup_{h < j} I_h) \cap S| < \eps|S|$ holds. Thus $\Pr[|T_i \cap S| \leq t] = \Pr\left[\sum \Delta_j \leq t\right] \leq \Pr\left[\sum X_j \leq t\right]$ for any $t < \eps|S|$. We now turn our attention to upper bounding $\Pr\left[\sum X_j \leq t\right]$.

By Corollary \ref{cor:new_swap_round}, for any constant $\lambda \in [0, \eps/12]$, %
\[ \mathbb{E}\left[e^{-\lambda(X_j-\mu)} | \mathcal{F}_{j-1}\right] \leq e^{6\lambda^2 \mu/\eps} \]
where $\mu = (1-\eps)|S_j|/\chimax^{(i)} = (1-\eps)^2|S|/\chimax^{(i)}$ and $\mathcal{F}_{j-1}$ contains all information in the algorithm's state immediately before sampling $I_j$. Thus for any constant $\lambda \in [0, \eps/12]$, we have
\begin{alignat}{1}
\mathbb{E}\Biggl[\exp\Biggl(-&\lambda\left(\sum_{j=1}^{c_i} X_j - (\eps-2\eps^2)|S|\right)\Biggr)\Biggr] \notag \\
&\leq \mathbb{E}\left[\exp\left(-\lambda\sum_{j=1}^{c_i} \left(X_j - \mu\right)\right)\right] \label{eq:mgf_conc_1} \\
&= \mathbb{E}\left[\prod_{j=1}^{c_i} \exp\left(-\lambda\left(X_j - \mu\right)\right)\right] \label{eq:mgf_conc_2} \\
&\leq \exp\left(6c_i\mu \lambda^2/\eps\right) \label{eq:mgf_conc_3} \\
&\leq \exp\left(6\lambda^2|S|\right) \label{eq:mgf_conc_4}
\end{alignat}

Line (\ref{eq:mgf_conc_1}) follows by $c_i \mu = \lceil \eps \chimax^{(i)} \rceil \mu \geq \left(\eps \chimax^{(i)}\right)(1-\eps)^2 |S|/\chimax^{(i)} \geq (\eps-2\eps^2)|S|$. Line (\ref{eq:mgf_conc_2}) follows by simplification. Line (\ref{eq:mgf_conc_3}) requires some explanation. Let $Z_j = \sum_{h=1}^j (X_h-\mu)$ for each $j \in [c_i]$. Then by the tower property of conditional expectation,

\[ \mathbb{E}\left[e^{-\lambda Z_j}\right] = \mathbb{E}\left[e^{-\lambda Z_{j-1}} \mathbb{E}\left[e^{-\lambda(X_j-\mu)} | \mathcal{F}_{j-1}\right]\right] \]\[\leq e^{6\lambda^2 \mu/\eps} \mathbb{E}\left[e^{-\lambda Z_{j-1}}\right]\]
for all $j \in [c_i]$. Iterating over $j = c_i, c_i-1, c_i-2, \dots, 1$ proves the desired result. Line (\ref{eq:mgf_conc_4}) follows by $c_i \mu =  \lceil \eps \chimax^{(i)} \rceil \mu \leq \left((1+\eps)\left(\eps\chimax^{(i)}\right)\right)\left((1-\eps)|S|/\chimax^{(i)}\right) \leq \eps|S|$, so $6c_i\mu/\eps \leq 6|S|$. Continuing, we have
\begin{equation*}
\begin{split}
\Pr\bigl[|T_i \cap S| & \leq \left(\eps-2\eps^2\right)|S|-t\bigr] \notag\\
&=\Pr\left[\sum_{j=1}^{c_i}\Delta_j \leq \left(\eps-2\eps^2\right)|S|-t\right] \\
&\leq \Pr\left[\sum_{j=1}^{c_i}X_j \leq \left(\eps-2\eps^2\right)|S|-t\right] \\
&\leq \exp\left(-\frac{t^2}{24|S|}\right) %
\end{split}
\end{equation*}
for all $t \in [0, \eps|S|]$. 
The final inequality follows by the fact that MGF bounds imply concentration bounds. More specifically, if a real-valued random variable $X$ satisfies $\mathbb{E}\left[e^{-\lambda X}\right] \leq e^{\sigma^2 \lambda^2 / 2}$ for all $\lambda \in [0, L]$ for fixed $L, \sigma^2 > 0$, then $\Pr\left[X \leq -t\right] \leq e^{-t^2/(2\sigma^2)}$ for all $t \in \left[0, L\sigma^2\right]$ (see \cite{Vershynin2018}). In our application of this fact, $X := \sum_{j=1}^{c_i}X_j - (\eps-2\eps^2)|S|$, $\sigma^2 := 12|S|$, and $L := \eps/12$. Thus
\begin{alignat}{1}
\Pr\bigl[|T_i & \cap S| \leq \left(\eps-q\eps^2\right)|S|\bigr] \notag\\
&\leq \Pr\left[|T_i \cap S| \leq \left(\eps-(q/2+2)\eps^2\right)|S|\right] \label{eq:mgf_conc_1''} \\
&= \Pr\left[|T_i \cap S| \leq \left(\eps-2\eps^2\right)|S| - q/2 \cdot \eps^2|S|\right] \label{eq:mgf_conc_2''} \\
&\leq \exp\left(-\frac{\left(q/2 \cdot \eps^2 |S|\right)^2}{24|S|}\right) \label{eq:mgf_conc_3''} \\
&\leq \exp\left(-q^2\eps^4|S|/100\right) \label{eq:mgf_conc_4''}
\end{alignat}

Line (\ref{eq:mgf_conc_1''}) follows by $q/2+2 \leq q$ when $q \geq 4$. Line (\ref{eq:mgf_conc_2''}) follows by simplification. Line (\ref{eq:mgf_conc_3''}) follows by application of our concentration bound when $t = q/2 \cdot \eps^2|S|$, which is a valid choice because $t = q/2 \cdot \eps^2|S| \leq \eps|S|/2 \leq \eps|S|$. Line (\ref{eq:mgf_conc_4''}) follows by simplification and $96 \leq 100$.
\end{proof}

With the concentration bound for individual subsets established, we now provide the global argument for the reduction of the chromatic number on each round. The strategy is to show that if the chromatic number does not drop sufficiently, there must exist a specific structural witness, a ``high-density flat", that was not well-covered by our sampled sets. By applying a union bound over the set of all such flats and utilizing the local guarantee from Lemma~\ref{lemma:asymp_subset}, we conclude that the chromatic number drops sufficiently with high probability.

\begin{proof}[Proof  of Lemma \ref{lemma:asymp_key}] Fix one of the matroids at the start of a round let $\chi$  denote its chromatic number. Define a \emph{high-density flat} as a flat $S$ in the matroid s.t. $|S| > \left(1-\eps+99\eps^2\right)\chi r(S)$ where $r$ is the rank function of the matroid. We first show that if the chromatic number of the restricted matroid at the end of the round is greater than $\left(1-\eps+100\eps^2\right)\chi$, then there exists a high-density flat $S$ in the start matroid s.t. $|T \cap S| < \left(\eps-99\eps^2\right) \chi r(S)$ where $T$ is the union of the sampled common independent sets during the round. We then union bound this event over all high-density flats $S$ and apply Lemma \ref{lemma:asymp_subset}.

First, suppose the chromatic number of the restricted matroid at the end of the round is greater than $\left(1-\eps+100\eps^2\right)\chi$. Then the restricted matroid contains a subset $S'$ s.t. $|S'| > \left(1-\eps+99\eps^2\right)\chi r(S')$ via Fact~\ref{onematcol}. Let $S$ be the span of $S'$ in the start matroid. Then $S$ is a high-density flat in the start matroid, because $S$ is a flat (since it is defined as the span of a collection of elements) and $|S| > \left(1-\eps+99\eps^2\right)\chi r(S)$ (since $|S| \geq |S'|$ and $r(S) = r(S')$). Further, $S' \subseteq S \setminus T$, giving $|T \cap S| = |S| - |S \setminus T| \leq |S| - |S'|$. Because $|S| \leq \chi r(S)$ via Fact~\ref{onematcol}, we have 
\begin{equation*}
\begin{split}
|S'| & > \left(1-\eps+99\eps^2\right)\chi r(S')
=\left(1-\eps+99\eps^2\right)\chi r(S) \\
&\geq \left(1-\eps+99\eps^2\right)|S|
\end{split}
\end{equation*}
which gives
\begin{equation*}
\begin{split}
|T \cap S| & \leq |S| - |S'| 
< |S| - \left(1-\eps+99\eps^2\right)|S| \\
&= \left(\eps-99\eps^2\right)|S|
\end{split}
\end{equation*}

Next, let $\mathcal{S}$ be the collection of all high-density flats in the start matroid, and let $\chi'$ be the chromatic number of the restricted matroid. Then
\begin{alignat}{1}
\Pr\bigl[\chi' & > \left(1-\eps+100\eps^2\right)\chi\bigr] \notag\\
&\leq \Pr\left[\exists S \in \mathcal{S} : |T \cap S| \leq \left(\eps-99\eps^2\right)|S|\right] \label{eq:asymp_1} \\
&\leq \sum_{S \in \mathcal{S}} \Pr\left[|T \cap S| \leq \left(\eps-99\eps^2\right)|S|\right] \label{eq:asymp_2} \\
&\leq \sum_{S \in \mathcal{S}} \exp\left(-98\eps^4|S|\right) \label{eq:asymp_3} \\
&= \sum_{m=1}^{n} \sum_{S \in \mathcal{S} : r(S) = m} \exp\left(-98\eps^4|S|\right) \label{eq:asymp_4} \\
&\leq \sum_{m=1}^{n} \sum_{S \in \mathcal{S} : r(S) = m} \exp\left(-49\eps^4 \chi m\right) \label{eq:asymp_5} \\
&\leq \sum_{m=1}^n n^{m} \exp\left(-49\eps^4 \chi m\right)  \label{eq:asymp_6} \\
&\leq \sum_{m=1}^n \exp(m \log n - 49 m \log n) \label{eq:asymp_7} \\
&\leq \sum_{m=1}^n n^{-20m} \label{eq:asymp_8} \\
&\leq n^{-10} \label{eq:asymp_9}
\end{alignat}

Line (\ref{eq:asymp_1}) follows from the previous paragraph. Line (\ref{eq:asymp_2}) follows by a union bound. Line (\ref{eq:asymp_3}) follows from Lemma \ref{lemma:asymp_subset} and $99^2/100 > 98$. Line (\ref{eq:asymp_4}) follows by partitioning the high-density flats by rank. Line (\ref{eq:asymp_5}) follows by $|S| > \left(1-\eps+99\eps^2\right) \chi m > \chi m/2$. Line (\ref{eq:asymp_6}) follows from the fact that a flat of rank $m$ is defined as the span of an independent set of size $m$, and so the number of flats of rank $m$ in a matroid with $n$ elements is at most $\binom{n}{m} \leq n^{m}$. Line (\ref{eq:asymp_7}) follows by $\chi \geq \eps \chimax \geq \log n/\eps^4$. Line (\ref{eq:asymp_8}) follows by computation. Line (\ref{eq:asymp_9}) follows from the fact that the sum contains $n$ terms of size at most $n^{-20}$.
\end{proof}

\subsubsection*{Runtime of Algorithm \ref{alg:peeling}.} Lastly, we remark on the runtime of Algorithm \ref{alg:peeling}. Phase 1 involves $\ell = O\left(\frac{\log(1/\eps)}{\eps}\right)$ rounds. In each round $i$, we call SwapRound $c_i = \lceil \eps \chimax^{(i)}\rceil = O(\eps n)$ times, and run for poly$(n)$ time to compute the chromatic numbers of $M_1^{(i)}, M_2^{(i)}$. The runtime of $\text{SwapRound}\left(M_1^{(i)}, M_2^{(i)}, 1/\chimax^{(i)}, \eps\right)$ is polynomial in $n, 1/\eps$ as discussed in the background section. Thus the runtime of Phase 1 is polynomial in $n, 1/\eps$. The runtime of Phase $2$ is polynomial in $n$, so the runtime of Algorithm \ref{alg:peeling} is polynomial in $n, 1/\eps$.

\appendix

\section{Proof of Theorem \ref{thm:largec_clean}}
\label{app:largec_techclean} %

If $\ve\geq 1$, then the $(1+\ve)\chimax$ bound follows from Theorem~\ref{twomatcol}, for any $\chimax$. So suppose $\ve\in(0,1)$.
Let $A_\eps$ be the randomized algorithm given by Theorem \ref{thm:largec_technical}. 
We take $C = 1000^5$, and the randomized algorithm in the  statement of Theorem~\ref{thm:largec_clean} to be $A_{\eps/1000}$. 
(More precisely, the randomized algorithm is $A_{\ve/1000}$, for $\ve\in(0,1)$, and the algorithm in Theorem~\ref{twomatcol} if $\ve\geq 1$).

We now show that this randomized algorithm satisfies the guarantee of Theorem \ref{thm:largec_clean} for all $\eps > 0$. We can focus on $\ve\in(0,1)$.
When $\chimax \geq C \log n / \eps^5 = \log n / (\eps / 1000)^5$, using Theorem~\ref{thm:largec_technical}, $A_{\eps/1000}$ produces a feasible coloring of $\Mint = M_1 \cap M_2$ using at most $(1+400(\eps/1000))\chimax \leq (1+\eps)\chimax$ colors with probability at least $1-1/n^2$.

Lastly, we can convert the probability of $1-1/n^2$ into a high probability statement by running the algorithm $O(\log n)$ times and outputting the minimum coloring. 
\hfill\qed

\end{document}